\documentclass[submission,copyright,creativecommons]{eptcs}
\providecommand{\event}{GandALF 2023} 

\usepackage{iftex}
\usepackage{amsmath, amsthm, amssymb}
\usepackage{xspace}
\usepackage{underscore}
\usepackage{enumitem}
\usepackage{comment}
\usepackage{hyperref}
\usepackage{graphicx}
\usepackage{pdfpages}
\usepackage{xcolor}
\ifpdf
  \usepackage{underscore}         
  \usepackage[T1]{fontenc}        
\else
  \usepackage{breakurl}           
\fi

\theoremstyle{plain}
\newtheorem{theorem}{Theorem}[section]
\newtheorem{proposition}[theorem]{Proposition}
\newtheorem{corollary}[theorem]{Corollary}
\newtheorem{lemma}[theorem]{Lemma}
\newtheorem{observation}[theorem]{Observation}

\theoremstyle{definition}
\newtheorem{definition}[theorem]{Definition}
\newtheorem{remark}[theorem]{Remark}
\newtheorem{question}[theorem]{Question}

\newcommand{\Lem}[1]{Lem.~\ref{#1}\xspace}

\newcommand{\Prop}[1]{Prop.~\ref{#1}\xspace}
\newcommand{\Thm}[1]{Thm.~\ref{#1}\xspace}

\newcommand{\wt}{\text{wt}}

\DeclareMathOperator{\Aut}{Aut}

\DeclareMathOperator{\rad}{Rad}
\DeclareMathOperator{\pker}{PKer}

\newcommand{\cc}[1]{\ensuremath{\mathsf{#1}}}
\newcommand{\algprobm}[1]{\textsc{#1}\xspace}

\newcommand{\Soc}{\text{Soc}}
\newcommand{\Fac}{\text{Fac}}

\title{On the Descriptive Complexity of Groups without Abelian Normal Subgroups \\
{\large (Extended Abstract)}}
\author{Joshua A. Grochow
\institute{Department of Computer Science University of Colorado Boulder, CO USA}
\institute{Department of Mathematics University of Colorado Boulder, CO USA}
\email{jgrochow@colorado.edu}
\and
Michael Levet
\institute{Department of Computer Science, College of Charleston, SC USA}
\email{levetm@cofc.edu}
}
\def\titlerunning{On the Descriptive Complexity of Groups without Abelian Normal Subgroups}
\def\authorrunning{J.A. Grochow \& M. Levet}
\begin{document}
\maketitle

\begin{abstract}
In this paper, we explore the descriptive complexity theory of finite groups by examining the power of the second Ehrenfeucht--Fra\"iss\'e bijective pebble game in Hella's (\textit{Ann. Pure Appl. Log.}, 1989) hierarchy. This is a Spoiler--Duplicator game in which Spoiler can place up to two pebbles each round. While it trivially solves graph isomorphism, it may be nontrivial for finite groups, and other ternary relational structures. We first provide a novel generalization of Weisfeiler--Leman (WL) coloring, which we call \textit{2-ary} WL. We then show that the 2-ary WL is equivalent to the second Ehrenfeucht--Fra\"iss\'e bijective pebble game in Hella's hierarchy. 

Our main result is that, in the pebble game characterization, only $O(1)$ pebbles and $O(1)$ rounds are sufficient to identify all groups without Abelian normal subgroups (a class of groups for which isomorphism testing is known to be in $\mathsf{P}$; 
Babai, Codenotti, \& Qiao, ICALP 2012). In particular, we show that within the first few rounds, Spoiler can force Duplicator to select an isomorphism between two such groups at each subsequent round. By Hella's results (\emph{ibid.}), this is equivalent to saying that these groups are identified by formulas in first-order logic with generalized 2-ary quantifiers, using only $O(1)$ variables and $O(1)$ quantifier depth. 
\end{abstract}

\section{Introduction}
\label{sec:introduction}
Descriptive complexity theory studies the relationship between the complexity of describing a given problem in some logic, and the complexity of solving the problem by an algorithm. When the problems involved are isomorphism problems, Immerman and Lander \cite{ImmermanLander1990} showed that complexity of a logical sentence describing the isomorphism type of a graph was essentially the same as the Weisfeiler--Leman coloring dimension of that graph, and the complexity of an Ehrenfeucht--Fra\"iss\'e pebble game (see also \cite{CFI}). 

It is a well-known open question whether there is a logic that exactly captures the complexity class $\textsf{P}$ on unordered (unlabeled) structures; on ordered structures such a logic was given by Immerman \cite{ImmermanPTime} and Vardi \cite{VardiPTime}. The difference between these two settings is essentially the \textsc{Graph Canonization} problem, whose solution allows one to turn an unordered graph into an ordered graph in an isomorphism-preserving way.

One natural approach in trying to capture $\textsf{P}$ on unordered structures is thus to attempt to extend first-order logic $\textsf{FO}$ by generalized quantifiers (c.f., Mostowski \cite{Mostowski1957} and Lindstrom \cite{Lindstrom}) in the hopes that the augmented logics can characterize finite graphs up to isomorphism, thus reducing the unordered case to the previously solved ordered case. A now-classical approach, initiated by Immerman \cite{ImmermanPTime}, was to augment fixed-point logic with counting quantifiers, which can be analyzed in terms of an equivalence induced by (variable confined) fragments of first-order logic with counting. However, Cai, F\"urer, \& Immerman \cite{CFI} showed that $\textsf{FO}+\textsf{LFP}$ plus counting does not capture $\textsf{P}$ on finite graphs. More generally, Flum \& Grohe have characterized when $\textsf{FO}$ plus counting captures $\textsf{P}$ on unordered structures \cite{GroheFlum}.

The approach of Cai, F\"urer, \& Immerman (ibid., see also \cite{ImmermanLander1990}) was to prove a three-way equivalence: between (1) counting logics, (2) the higher-dimensional Weisfeiler--Leman coloring procedure, and (3) Ehrenfeucht--Fra\"iss\'e pebble games. Ehrenfeucht--Fra\"iss\'e pebble games \cite{Ehrenfeucht, Fraisse} have long been an important tool in proving the inexpressibility of certain properties in various logics; in this case, they used such games to show that the logics could not express the difference between certain pairs of non-isomorphic graphs. Consequently, Cai, F\"urer, \& Immerman ruled out the Weisfeiler--Leman algorithm as a polynomial-time isomorphism test for graphs, which resolved a long-standing open question in isomorphism testing. Nonetheless, the Weisfeiler--Leman coloring procedure is a key subroutine in many algorithms for \textsc{Graph Isomorphism}, including Babai's quasi-polynomial-time algorithm \cite{BabaiGraphIso}. It is thus interesting to study its properties and its distinguishing power.

While the result of Cai, F\"urer, \& Immerman ruled out Weisfeiler--Leman as a polynomial-time isomorphism test for graphs, for \emph{groups} it remains an interesting open question. The general WL procedure for groups was introduced by Brachter \& Schweitzer \cite{WLGroups} and has been studied in several papers since then \cite{BrachterSchweitzerWLLibrary,GrochowLevetWL}. Outside the scope of WL, it is known that \textsc{Group Isomorphism} is $\textsf{AC}^{0}$-reducible to \textsc{Graph Isomorphism}, and there is no $\textsf{AC}^0$ reduction in the opposite direction \cite{ChattopadhyayToranWagner}. For this and other reasons, group isomorphism is believed to be the easier of the two problems, so it is possible that WL---and more generally, tools from descriptive complexity---could yield stronger results for groups than for graphs.

On graphs, which are binary relational structures, if Spoiler is allowed to pebble two elements per turn, then Spoiler can win on any pair of non-isomorphic graphs. However, groups are ternary relational structures (the relation is $\{(a,b,c) : ab=c\}$), so such a game may yield nontrivial insights into the descriptive complexity of finite groups. Hella \cite{Hella1989,Hella1993} introduced such games in a more general context, and showed that allowing Spoiler to pebble $q$ elements per round corresponded to the generalized $q$-ary quantifiers of Mostowski \cite{Mostowski1957} and Lindstrom \cite{Lindstrom}. When $q=1$, Hella shows that this pebble game is equivalent in power to the $\textsf{FO}$ plus counting logics mentioned above. Our focus in this paper is to study the power of the $q=2$-ary game for identifying finite groups. 

\noindent \\ \textbf{Main Results.} In this paper, we initiate the study of Hella's $2$-ary Ehrenfeucht--Fra\"iss\'e-style pebble game, in the setting of groups. Our main result is that this pebble game efficiently characterizes isomorphism in a class of groups for which isomorphism testing is known to be in $\mathsf{P}$, but only by quite a nontrivial algorithm (see remark below). The full version of this paper appears on arXiv \cite{grochow2022descriptive}.

\vskip 5pt
\begin{theorem} \label{thm:SemisimpleMain}
Let $G$ be a group with no Abelian normal subgroups (a.k.a. Fitting-free or semisimple), and let $H$ be arbitrary. If $G \not \cong H$, then Spoiler has a winning strategy in the Ehrenfeucht--Fra\"iss\'e game at the second level of Hella's hierarchy, using $9$ pebbles and $O(1)$ rounds.
\end{theorem}

In proving \Thm{thm:SemisimpleMain}, we show that with the use of only a few pebbles, Spoiler can effectively force Duplicator to select an isomorphism of $G$ and $H$. We contrast this with the setting of Weisfeiler--Leman (which is equivalent to the $1$-ary pebble game), for which the best upper bound we have on the WL-dimension is the trivial bound of $\log n$. Furthermore, we do not have any lower bounds on the WL-dimension for semisimple groups.

\begin{remark}
Every group $G$ can be written as an extension of its solvable radical $\rad(G)$ by the quotient $G/\rad(G)$, which does not have Abelian normal subgroups. As such, the latter class of groups is quite natural, both group-theoretically and computationally. Computationally, it has been used in algorithms for general finite groups both in theory (e.g., \cite{Babai1999GroupsSA,BabaiBealsSeress}) and in practice (e.g., \cite{CH03}). Isomorphism testing in this family of groups can be solved efficiently in practice \cite{CH03}, and is known to be in $\cc{P}$ through a series of two papers \cite{BCGQ, BCQ}.


\end{remark}

In Section~\ref{sec:coloring} we also complete the picture by giving a Weisfeiler--Leman-style coloring procedure and showing that it corresponds precisely to Hella's $q$-ary pebble games and $q$-ary generalized Lindstrom quantifiers \cite{Lindstrom}. When the groups are given by their multiplication tables, this procedure runs in time $n^{\Theta(\log^{2} n)}$ by reduction to \algprobm{Graph Isomorphism}. We note that Hella's results deal with infinitary logics \cite{Hella1989,Hella1993}. However, as we are dealing with finite groups, the infinitary quantifiers and connectives are not necessary (see the discussion in \cite{Hella1993}, right above Theorem 5.3).

\noindent \\ \textbf{Further Related Work.} Despite the fact that Weisfeiler--Leman is insufficient to place \algprobm{Graph Isomorphism} (\algprobm{GI}) into $\textsf{PTIME}$, it remains an active area of research. For instance, Weisfeiler--Leman is a key subroutine in Babai's quasipolynomial-time \algprobm{GI} algorithm \cite{BabaiGraphIso}. Furthermore, Weisfeiler--Leman has led to advances in simultaneously developing both efficient isomorphism tests and the descriptive complexity theory for finite graphs- see for instance, \cite{GroheBook,GroheVerbitsky,KieferMcKay,KieferPonomarenkoSchweitzer,grohe2019linear,grohe2019rankwidth, grohe2021logarithmic,KieferSchweitzerSelman,ajtaifagin1990,ARORA199797,Rossman2009EhrenfeuchtFrassGO}. Weisfeiler--Leman also has close connections to the Sherali--Adams hierarchy in linear programming  \cite{GroheOttoLinearEquations}.

The complexity of the \algprobm{Group Isomorphism} (\algprobm{GpI}) problem is a well-known open question. In the Cayley (multiplication) table model, $\algprobm{GpI}$ belongs to $\textsf{NP} \cap \textsf{coAM}$. The generator-enumerator algorithm, attributed to Tarjan in 1978 \cite{MillerTarjan}, has time complexity $n^{\log_{p}(n) + O(1)}$, where $n$ is the order of the group and $p$ is the smallest prime dividing $n$. This bound has escaped largely unscathed: Rosenbaum \cite{Rosenbaum2013BidirectionalCD} (see \cite[Sec. 2.2]{GR16}) improved this to $n^{(1/4)\log_p(n) + O(1)}$. And even the impressive body of work on practical algorithms for this problem, led by Eick, Holt, Leedham-Green and O'Brien (e.\,g., \cite{BEO02,ELGO02,BE99,CH03}) still results in an $n^{\Theta(\log n)}$-time algorithm in the general case (see \cite[Page 2]{WilsonSubgroupProfiles}). In the past several years, there have been significant advances on algorithms with worst-case guarantees on the serial runtime for special cases of this problem including Abelian groups \cite{Kavitha,Vikas,Savage}, direct product decompositions \cite{WilsonDirectProductsArxiv,KayalNezhmetdinov}, groups with no Abelian normal subgroups \cite{BCGQ,BCQ}, coprime and tame group extensions \cite{Gal09,QST11,BQ, GQ15}, low-genus $p$-groups and their quotients \cite{LW12,BMWGenus2}, Hamiltonian groups \cite{DasSharma}, and groups of almost all orders \cite{DietrichWilson}.

Key motivation for $\algprobm{GpI}$ is due to its close relation to $\algprobm{GI}$. In the Cayley (verbose) model, $\algprobm{GpI}$ reduces to $\algprobm{GI}$ \cite{ZKT}, 
while $\algprobm{GI}$ reduces to the succinct $\algprobm{GpI}$ problem \cite{Heineken1974TheOO,Mekler} (recently simplified \cite{HeQiao}). In light of Babai's breakthrough result that $\algprobm{GI}$ is quasipolynomial-time solvable \cite{BabaiGraphIso}, $\algprobm{GpI}$ in the Cayley model is a key barrier to improving the complexity of $\algprobm{GI}$. Both verbose $\algprobm{GpI}$ and $\algprobm{GI}$ are considered to be candidate $\textsf{NP}$-intermediate problems, that is, problems that belong to $\textsf{NP}$, but are neither in $\textsf{P}$ nor $\textsf{NP}$-complete \cite{Ladner}. There is considerable evidence suggesting that $\algprobm{GI}$ is not $\textsf{NP}$-complete \cite{Schoning,BuhrmanHomer,ETH,BabaiGraphIso,GILowPP,ArvindKurur}. As verbose $\algprobm{GpI}$ reduces to $\algprobm{GI}$, this evidence also suggests that $\algprobm{GpI}$ is not $\textsf{NP}$-complete. It is also known that $\algprobm{GI}$ is strictly harder than $\algprobm{GpI}$ under $\textsf{AC}^{0}$ reductions \cite{ChattopadhyayToranWagner}. 

While the descriptive complexity of graphs has been extensively studied, the work on the descriptive complexity of groups is scant compared to the algorithmic literature on \algprobm{Group Isomorphism} (\algprobm{GpI}). There has been work relating first order logics and groups \cite{FiniteGroupsFOL}, as well as work examining the descriptive complexity of finite abelian groups \cite{DescriptiveComplexityAbelianGroups}. Recently, Brachter \& Schweitzer \cite{WLGroups} introduced three variants of Weisfeiler--Leman for groups, including corresponding logics and pebble games. These pebble games correspond to the first level of Hella's hierarchy \cite{Hella1989,Hella1993}. In particular, Brachter \& Schweitzer showed that $3$-dimensional Weisfeiler--Leman can distinguish $p$-groups arising from the CFI graphs \cite{CFI} via Mekler's construction \cite{Mekler}, suggesting that $\textsf{FO} + \textsf{LFP} + \textsf{C}$ may indeed capture $\textsf{PTIME}$ on groups. Determining whether even $o(\log n)$-dimensional Weisfeiler--Leman can resolve \algprobm{GpI} is an open question.

The use on Weisfeiler--Leman for groups is quite new. To the best of our knowledge, using Weisfeiler--Leman for \algprobm{Group Isomorphism} testing was first attempted by Brooksbank, Grochow, Li, Qiao, \& Wilson \cite{BGLQW}. Brachter \& Schweitzer \cite{WLGroups} subsequently introduced three variants of Weisfeiler--Leman for groups that more closely resemble that of graphs. In particular, Brachter \& Schweitzer \cite{WLGroups}  characterized their algorithms in terms of logics and Ehrenfeucht--Fra\"iss\'e pebble games. The relationship between the works of Brachter \& Schweitzer and Brooksbank, Grochow, Li, Qiao, \& Wilson \cite{BGLQW} is an interesting question.

In subsequent work, Brachter \& Schweitzer \cite{BrachterSchweitzerWLLibrary} further developed the descriptive complexity of finite groups. They showed in particular that low-dimensional Weisfeiler--Leman can detect key group-theoretic invariants such as composition series, radicals, and quotient structure. Furthermore, they also showed that Weisfeiler--Leman can identify direct products in polynomial-time, provided it can also identify the indecomposable direct factors in polynomial-time. Grochow \& Levet \cite{GrochowLevetWL} extended this result to show that Weisfeiler--Leman can compute direct products in parallel, provided it can identify each of the indecomposable direct factors in parallel. Additionally, Grochow \& Levet showed that constant-dimensional Weisfeiler--Leman can in a constant number of rounds identify coprime extensions $H \ltimes N$, where the normal Hall subgroup $N$ is Abelian and the complement $H$ is $O(1)$-generated. This placed isomorphism testing into $\textsf{L}$; the previous bound for isomorphism testing in this family was $\textsf{P}$ \cite{QST11}. Grochow \& Levet also ruled out $\textsf{FO} + \textsf{LFP}$ as a candidate logic for capturing $\textsf{PTIME}$ on finite groups, by showing that the count-free Weisfeiler--Leman algorithm cannot even identify Abelian groups in polynomial-time.

\section{Preliminaries}





We recall the bijective pebble game of Hella \cite{Hella1989,Hella1993}, in the context of WL on graphs as that is likely more familiar to more readers. This game is often used to show that two graphs $X$ and $Y$ cannot be distinguished by $k$-WL. The game is an Ehrenfeucht--Fra\"iss\'e game, with two players: Spoiler and Duplicator. Each graph begins with $k+1$ pebbles, $p_1, \dotsc, p_{k+1}$ for $X$ and $p_1^{\prime}, \dotsc, p_{k+1}^{\prime}$ for $Y$, which are placed beside the graphs. Each round proceeds as follows.
\begin{enumerate}
\item Spoiler chooses $i \in [k+1]$, and picks up pebbles $p_i, p^{\prime}_i$.
\item We check the winning condition, which will be formalized later.\footnote{In the literature, some authors check the winning condition at this point, and others check the winning condition at the end of each round. The choice merely has the effect of changing the number of required pebbles by at most $1$ in ordinary WL, or at most $q$ in the $q$-ary version, and changing the number of rounds by at most 1. We have chosen this convention for consistency with other works on WL specific to groups \cite{WLGroups, BrachterSchweitzerWLLibrary, GrochowLevetWL}.}
\item Duplicator chooses a bijection $f : V(X) \to V(Y)$.
\item Spoiler places $p_{i}$ on some vertex $v \in V(X)$. Then $p_{i}^{\prime}$ is placed on $f(v)$. 
\end{enumerate} 

In a given round, let $v_{1}, \ldots, v_{m}$ be the vertices of $X$ pebbled at the end of step 1 (in the list above), and let $v_{1}^{\prime}, \ldots, v_{m}^{\prime}$ be the corresponding pebbled vertices of $Y$. Spoiler wins precisely if the map $v_{\ell} \mapsto v_{\ell}^{\prime}$ is not an isomorphism of the induced subgraphs $X[\{v_{1}, \ldots, v_{m}\}]$ and $Y[\{v_{1}^{\prime}, \ldots, v_{m}^{\prime}\}]$. Otherwise, at that point, Duplicator wins the game. Spoiler wins, by definition, at round $0$ if $X$ and $Y$ do not have the same number of vertices. We note that $X$ and $Y$ are not distinguished by the first $r$ rounds of $k$-WL if and only if Duplicator wins the first $r$ rounds of the $(k+1)$-pebble game \cite{Hella1989,Hella1993,CFI}. 

Hella \cite{Hella1989,Hella1993} exhibited a hierarchy of pebble games where, for $q \geq 1$, Spoiler could pebble a sequence of $1 \leq j \leq q$ elements $(v_{1}, \ldots, v_{j}) \mapsto (f(v_{1}), \ldots, f(v_{j}))$ in a single round; more formally, following the description above, in step 1, Spoiler picks up $q$ pebbles $p_{i_1}, \dotsc, p_{i_q}$ and their partners $p'_{i_1}, \dotsc, p'_{i_q}$, with step 4 changed accordingly. The case of $q = 1$ corresponds to the case of Weisfeiler--Leman. As remarked by Hella \cite[p.~6, just before \S 4]{Hella1993}, the $q$-ary game immediately identifies all relational structures of arity $\leq q$. For example, the $q=2$ game on graphs solves $\algprobm{GI}$: for if two graphs $X$ and $Y$ are non-isomorphic, then any bijection $f : V(X) \to V(Y)$ that Duplicator selects must map an adjacent pair of vertices $u,v$ in $X$ to a non-adjacent pair $f(u), f(v)$ in $Y$ or vice-versa. Spoiler immediately wins by pebbling $(u, v) \mapsto (f(u), f(v))$. However, as groups are ternary relational structures (the relation being $\{(a,b,c) : a,b,c \in G, ab=c\}$), the $q=2$ case can, at least in principle, be non-trivial on groups. 

Brachter \& Schweitzer \cite{WLGroups} adapted Hella's \cite{Hella1989,Hella1993} pebble games in the $q = 1$ case to the setting of groups, obtaining three different versions. Their Version III involves reducing to graphs and playing the pebble game on graphs, so we don't consider it further here. Versions I and II are both played on the groups $G$ and $H$ directly.

Both versions are played identically as for graphs, with the only difference being the winning condition. We recall the following standard definitions in order to describe these winning conditions.

\vskip 5pt
\begin{definition}
Let $G,H$ be two groups. Given $k$-tuples $\overline{g} = (g_1, \dotsc, g_k) \in G^k$ and $\overline{h} = (h_1, \dotsc, h_k) \in H^k$, we say $(\overline{g}, \overline{h})$ ...
\begin{enumerate}
\item ...\emph{gives a well-defined map} if $g_i = g_j \Leftrightarrow h_i = h_j$ for all $i \neq j$;

\item ...are \emph{partially isomorphic} or \emph{give a partial isomorphism} if they give a well-defined map, and for all $i,j,k$ we have $g_i g_j = g_k \Leftrightarrow h_i h_j = h_k$;

\item ...are \emph{marked isomorphic} or \emph{give a marked isomorphism} if it gives a well-defined map, and the map extends to an isomorphism $\langle g_1, \dotsc, g_k \rangle \to \langle h_1, \dotsc, h_k \rangle$.
\end{enumerate}
\end{definition}

Let $v_{1}, \ldots, v_{m}$ be the group elements of $G$ pebbled at the end of step 1, and let $v_{1}^{\prime}, \ldots, v_{m}^{\prime}$ be the corresponding pebbled vertices of $H$. In Version I, Spoiler wins precisely if $(\overline{v}, \overline{v}^{\prime})$ does not give a partial isomorphism, and in Version II Spoiler wins precisely if $(\overline{v}, \overline{v}^{\prime})$ does not give a marked isomorphism.

Both Versions I and II may be generalized to allow Spoiler to pebble up to $q$ group elements at a single round, for some $q \geq 1$. Mimicking the proof above for $q=2$ for graphs, we have that $q = 3$ is sufficient to solve $\algprobm{GpI}$ in a single round. The distinguishing power, however, of the $q = 2$ game for groups remains unclear, and is the main subject of this paper. As we are interested in the round complexity, we introduce the following notation. 

\vskip 5pt
\begin{definition}[Notation for pebbles, rounds, arity, and WL version]
Let $k \geq 2, r \geq 1$, $q \geq 1$, and $J \in \{I, II\}$. Denote $(k,r)$-WL$^q_J$ to be the $k$-pebble, $r$-round, $q$-ary Version $J$ pebble game.
\end{definition}

We refer to $q$ as the \emph{arity} of the pebble game, as it corresponds to the arity of generalized quantifiers\footnote{As our focus in this paper is not on the viewpoint of generalized quantifiers, we refer the reader to \cite{Hella1989} for details.} in a logic whose distinguishing power is equivalent to that of the game:

\begin{remark}[Equivalence with logics with generalized $2$-ary quantifiers]
Hella \cite{Hella1989} describes the game (essentially the same as our description, but with no restriction on number of pebbles, and a transfinite number of rounds) for general $q$ at the bottom of p.~245, for arbitrary relational structures. We restrict to the case of $q = 2$, a finite number of pebbles and rounds, and the (relational) language of groups. Hella proves that this game is equivalent to first-order logic with arbitrary $q$-ary equantifiers in \cite[Thm.~2.5]{Hella1989}.
\end{remark}

\vskip 5pt
\begin{observation}
In the $2$-ary pebble game, we may assume that Duplicator selects bijections that preserve inverses. 
\end{observation}

\begin{proof}
Suppose not. First, Duplicator must select bijections that preserve the identity, for if not, Spoiler pebbles $1_G \mapsto f(1) \neq 1_H$ and wins immediately. Next, let $f : G \to H$ be a bijection such that $f(g^{-1}) \neq f(g)^{-1}$. Spoiler pebbles $(g, g^{-1}) \mapsto (f(g), f(g^{-1}))$. Now $gg^{-1} = 1$, while $f(g)f(g^{-1}) \neq 1$. So Spoiler wins. \end{proof}

\vskip 5pt
\noindent We frequently use this observation without mention.

\section{Higher-arity Weisfeiler-Leman-style coloring corresponding to higher arity pebble games} \label{sec:coloring}
%

Given a $k$-tuple $\overline{x} = (x_1, \dotsc, x_k) \in G^k$, a pair of distinct indices $i,j \in [k]$, and a pair of group elements $y, z$, we define $\overline{x}_{(i,j) \leftarrow (y,z)}$ to be the $k$-tuple $\overline{x}'$ that agrees with $\overline{x}$ on all indices besides $i,j$, and with $x_i' = y, x_j' = z$. If $i=j$, we require $y=z$, and we denote this $\overline{x}_{i \leftarrow y}$. 

Finally, two graphs $\Gamma_1, \Gamma_2$, with edge-colorings $c_i \colon E(\Gamma_i) \to C$ to some color set $C$ (for $i=1,2$) are color isomorphic if there is a graph isomorphism $\varphi \colon V(\Gamma_1) \to V(\Gamma_2)$ that also preserves colors, in the sense that $c_1((u,v)) = c_2((\varphi(u), \varphi(v))$ for all edges $(u,v) \in E(\Gamma_1)$.

\begin{definition}[2-ary $k$-dimensional Weisfeiler-Leman coloring]
Let $G,H$ be two groups of the same order, let $k \geq 1$. For all $k$-tuples $\overline{x}, \overline{y} \in G^k \cup H^k$:
\begin{itemize}
\item (Initial coloring, Version I) $\chi^{2,I}_0(\overline{x}) = \chi^{2,I}_0(\overline{y})$ iff $\overline{x}, \overline{y}$ are partially isomorphic.

\item (Initial coloring, Version II) $\chi^{2,II}_0(\overline{x}) = \chi^{2,II}_0(\overline{y})$ iff $\overline{x}, \overline{y}$ have the same marked isomorphism type.

\item (Color refinement) Given a coloring $\chi \colon G^k \cup H^k \to C$, the color refinement operator $R$ defines a new coloring $R(\chi)$ as follows. For each $k$-tuple $\overline{x} \in G^k$ (resp., $H^k$), we define an edge-colored graph $\Gamma_{\overline{x},\chi,i,j}$. If $i=j$, it is the graph on vertex set $V(\Gamma_{\overline{x},\chi,i,i}) = G$ (resp., $H$) with all self-loops and no other edges, where the color of each self-loop $(g,g)$ is $\chi(\overline{x}_{i \leftarrow g})$. If $i \neq j$, it is the complete directed graph with self-loops on vertex set $G$ (resp., $H$), where the color of each edge $(y,z)$ is $\chi(\overline{x}_{(i,j) \leftarrow (y,z)})$. For an edge-colored graph $\Gamma$, we use $[\Gamma]$ to denote its edge-colored isomorphism class. We then define
\[
R(\chi)(\overline{x}) = \left(\chi(\overline{x}); [\Gamma_{\overline{x},\chi,1,1}], [\Gamma_{\overline{x},\chi,1,2}], \dotsc, [\Gamma_{\overline{x},\chi,k-1,k}], [\Gamma_{\overline{x},\chi,k,k}]\right).
\]
That is, the new color consists of the old color, as well as the tuple of $\binom{k+1}{2}$ edge-colored isomorphism types of the graphs $\Gamma_{\overline{x}, \chi, i, j}$.
\end{itemize}
The refinement operator may be iterated: $R^t(\chi) := R(R^{t-1}(\chi))$, and we define the \emph{stable refinement} of $\chi$ as $R^t(\chi)$ where the partition induced by $R^t(\chi)$ on $G^k \cup H^k$ is the same as that induced by $R^{t+1}(\chi)$. We denote the stable refinement by $R^\infty(\chi)$.

Finally, for $J \in \{I, II\}$ and all $r \geq 0$, we define $\chi^{2,J}_{r+1} = R(\chi^{2,J}_r)$, and $\chi^{2,J}_{\infty} := R^\infty(\chi^{2,J}_0)$.
\end{definition}

\vskip 5pt
\begin{remark}
Brachter \& Schweitzer \cite{WLGroups} introduced Versions I and II of $1$-ary WL, which are equivalent up to a small additive constant in the WL-dimension \cite{WLGroups} and $O(\log n)$ rounds \cite{GrochowLevetWL}. For the purpose of comparison, we introduce Versions I and II of $2$-ary WL. We will see later that only one additional round suffices in the $2$-ary case (see \Thm{thm:rounds}). The differences in Versions I and II of WL (both the $1$-ary and $2$-ary variants) arise from whether the group is viewed as a structure with a ternary relational structure (Version I) or as a structure with a binary function (Version II).
\end{remark}

\begin{remark}
Since it was one of our stumbling blocks in coming up with this generalized coloring, we clarify here how this indeed generalizes the usual 1-ary WL coloring procedure. In the 1-ary ``oblivious'' $k$-WL procedure (see \cite[\S 5]{grohe}, equivalent to ordinary WL), the color of a $k$-tuple $\overline{x}$ is refined using its old color, together with a $k$-tuple of multisets 
\[
(\{\!\{\chi(x_{1 \leftarrow y}) : y \in G \}\!\}, \{\!\{\chi(x_{2 \leftarrow y}) : y \in G \}\!\}, \dotsc, \{\!\{\chi(x_{k \leftarrow y}) : y \in G \}\!\}).
\]
For each $i$, note that two multisets $\{\!\{ \chi(x_{i \leftarrow y}) : y \in G \}\!\}$ and $\{\!\{ \chi(x'_{i \leftarrow y}) : y \in G \}\!\}$ are equal iff the graphs $\Gamma_{\overline{x}, \chi, i, i}$ and $\Gamma_{\overline{x}', \chi, i, i}$ are color-isomorphic. That is, edge-colored graphs with only self-loops and no other edges are essentially the same, up to isomorphism, as multisets. Our procedure generalizes this by also considering graphs with other edges, which (as we'll see in the proof of equivalence, which will appear in the full version) are used to encode the choice of 2 simultaneous pebbles by Spoiler in each move of the game.
\end{remark}

\begin{theorem} \label{thm:coloring}
Let $G, H$ be two groups of order $n$, with $\overline{x} \in G^k, \overline{y} \in H^k$. Starting from the initial pebbling $x_i \mapsto y_i$ for all $i=1,\dotsc,k$, Spoiler has a winning strategy in the $k$-pebble, $r$-round, 2-ary Version $J$ pebble game (for $J \in \{I, II\}$) iff $\chi^{2,J}_r(\overline{x}) \neq \chi^{2,J}_r(\overline{y})$.
\end{theorem}

\begin{proof}
To appear in the full version.
\end{proof}

\begin{corollary}
For two groups $G,H$ of the same order and any $k \geq 1$, the following are equivalent:
\begin{enumerate}
\item The 2-ary $k$-pebble game does not distinguish two groups $G,H$
\item The multisets of stable colors on $G^k$ and $H^k$ are the same, that is, $\{\!\{ \chi^{2,J}_\infty(\overline{x}) : \overline{x} \in G^k \}\!\} = \{\!\{ \chi^{2,J}_\infty(\overline{y}) : \overline{y} \in H^k\}\!\}$
\item $\chi^{2,J}_\infty((1_G, 1_G, \dotsc, 1_G)) = \chi^{2,J}_\infty((1_H, \dotsc, 1_H))$.
\end{enumerate}
\end{corollary}

The analogous result holds in the $q=1$ case, going back to \cite{WLGroups}. 

\begin{proof}
To appear in the full version.
\end{proof}

\begin{remark}
For arbitrary relational structures with relations of arity $a+1$, the $a$-order pebble game may still be nontrivial, as pointed out in Hella \cite[p.~6, just before \S 4]{Hella1993}. Our coloring procedure generalizes in the following way to this more general setting, and the proof of the equivalence between the coloring procedure and Hella's pebble game is the same as the above, \emph{mutatis mutandis}. The main change is that for an $a$-th order pebble game, instead of just considering a graph on edges of size 1 (when $i=j$) or 2 (when $i \neq j$), we consider an $a'$-uniform directed hypergraph, where each hyperedge consists of a list of $a'$ vertices, for all $1 \leq a' \leq a$. This gives a coloring equivalent of the logical and game characterizations provided by Hella; this trifecta is partly why we feel it is justified to call this a ``higher-arity Weisfeiler--Leman'' coloring procedure. 

We note that there has been some work on equivalences with specific binary and higher-arity quantifiers: see for instance, the invertible map game of Dawar \& Holm \cite{DawarHolm} which generalizes rank logic, in which Spoiler can place multiple pebbles, but the bijections Duplicator selects must satisfy additional structure. Subsequently, Dawar \& Vagnozzi \cite{DawarVagnozzi} provided a generalization of Weisfeiler--Leman that further subsumes the invertible map game. We note that Dawar \& Vagnozzi's ``$WL_{k,r}$'', although it looks superficially like our $r$-ary $k$-WL, is in fact quite different: in particular, their refinement step ``flattens'' a multiset of multisets into its multiset union, which loses information compared to our 2-ary (resp., $r$-ary) game; indeed, they show that their WL$_{*,r}$ is equivalent to ordinary (1-ary) WL for any fixed $r$, whereas already 2-ary WL can solve GI. In general, the relationship between Hella's $2$-ary game and the works of Dawar \& Holm and Dawar \& Vagnozzi remains open. 

\end{remark}

\subsection{Equivalence between 2-ary $(k,r)$-WL Versions I and II} \label{sec:equiv12}
In this section we show that, up to additive constants in the number of pebbles and rounds, 2-ary WL Versions I and II are equivalent in their distinguishing power. For two different WL versions $W,W'$, we write $W \preceq W'$ to mean that if $W$ distinguishes two groups $G$ and $H$, then so does $W'$.

\begin{theorem} \label{thm:rounds}
Let $k \geq 2, r \geq 1$. We have that:
\[
(k,r)\text{-WL}^2_{I} \preceq (k,r)\text{-WL}^2_{II} \preceq (k+2, r+1)\text{-WL}^2_{I}.
\]
\end{theorem}


\begin{proof}
To appear in the full version.
\end{proof}

\section{Descriptive Complexity of Semisimple Groups} \label{SectionSemisimple}

In this section, we show that the $(O(1), O(1))$-WL$_{II}^{2}$ pebble game can identify groups with no Abelian normal subgroups,\footnote{In many places, we will use $O(1)$ for number of pebbles or rounds; we believe all of these can be replaced with particular numbers by a straightforward, if tedious, analysis of our proofs. However, since our focus is on the fact that these numbers are constant rather than on the exact values, we use the $O(1)$ notation.} also known as semisimple groups. We begin with some preliminaries.

\subsection{Preliminaries}
Semisimple groups are motivated by the following characteristic filtration:
\[
1 \leq \rad(G) \leq \Soc^{*}(G) \leq \text{PKer}(G) \leq G,
\]

\noindent which arises in the computational complexity community where it is known as the Babai--Beals filtration \cite{Babai1999GroupsSA}, as well as in the development of practical algorithms for computer algebra systems (c.f., \cite{CH03}). We now explain the terms of this chain. Here, $\rad(G)$ is the \textit{solvable radical}, which is the unique maximal solvable normal subgroup of $G$; recall that a group $N$ is solvable if the sequence $N^{(0)} := N$, $N^{(i)} = [N^{(i-1)}, N^{(i-1)}]$ terminates in the trivial group after finitely many steps, and $[A,B]$ denotes the subgroup generated by $\{aba^{-1}b^{-1} : a \in A, b \in B \}$. The socle of a group, denoted $\Soc(G)$, is the subgroup generated by all the minimal normal subgroups of $G$. $\Soc^{*}(G)$ is the preimage of the socle $\Soc(G/\rad(G))$ under the natural projection map $\pi : G \to G/\rad(G)$. To define $\text{PKer}$, we start by examining the action on $\Soc(G / \rad(G)) \cong \Soc^*(G) / \rad(G)$ that is induced by the action of $G$ on $\Soc^*(G)$ by conjugation. As $\Soc^{*}(G)/\rad(G) \cong \Soc(G/\rad(G))$ is the direct product of finite, non-Abelian simple groups $T_{1}, \ldots, T_{k}$, this action permutes the $k$ simple factors, yielding a homomorphism $\varphi : G \to S_{k}$. The kernel of this action is denoted $\text{PKer}(G)$.

When $\rad(G)$ is trivial, $G$ has no Abelian normal subgroups (and vice versa). We refer to such groups as \textit{semisimple} (following \cite{BCGQ, BCQ}) or trivial-Fitting (following \cite{CH03}). As a semisimple group $G$ has no Abelian normal subgroups, we have that $\Soc(G)$ is the direct product of non-Abelian simple groups.  As the conjugation action of $G$ on $\Soc(G)$ permutes the direct factors of $\Soc(G)$, there exists a faithful permutation representation $\alpha : G \to G^{*} \leq \Aut(\Soc(G))$. $G$ is determined by $\Soc(G)$ and the action $\alpha$. Let $H$ be a semisimple group with the associated action $\beta : H \to \text{Aut}(\Soc(H))$. We have that $G \cong H$ precisely if $\Soc(G) \cong \Soc(H)$ via an isomorphism that makes $\alpha$ equivalent to $\beta$ in the sense introduced next. 

We now introduce the notion of permutational isomorphism, which is our notion of equivalence for $\alpha$ and $\beta$. Let $A$ and $B$ be finite sets, and let $\pi : A \to B$ be a bijection. For $\sigma \in \text{Sym}(A)$, let $\sigma^{\pi} \in \text{Sym}(B)$ be defined by $\sigma^{\pi} := \pi^{-1}\sigma \pi$. For a set $\Sigma \subseteq \text{Sym}(A)$, denote $\Sigma^{\pi} := \{ \sigma^{\pi} : \sigma \in \Sigma\}$. Let $K \leq \text{Sym}(A)$ and $L \leq \text{Sym}(B)$ be permutation groups. A bijection $\pi : A \to B$ is a \textit{permutational isomorphism} $K \to L$ if $K^{\pi} = L$.

The following lemma, applied with $R = \Soc(G)$ and $S = \Soc(H)$, gives a precise characterization of semisimple groups in terms of the associated actions.      
 
\begin{lemma}[{\cite[Lemma 3.1]{BCGQ}, cf. \cite[\S 3]{CH03}}] \label{CharacterizeSemisimple}
Let $G$ and $H$ be groups, with $R \triangleleft G$ and $S \triangleleft H$ groups with trivial centralizers. Let $\alpha : G \to \Aut(R)$ and $\beta : H \to \Aut(S)$ be faithful permutation representations of $G$ and $H$ via the conjugation action on $R$ and $S$, respectively. Let $f : R \to S$ be an isomorphism. Then $f$ extends to an isomorphism $\hat{f} : G \to H$ if and only if $f$ is a permutational isomorphism between $G^{*} = Im(\alpha)$ and $H^{*} = Im(\beta)$; and if so, $\hat{f} = \alpha f^{*} \beta^{-1}$, where $f^{*} :  G^{*} \to H^{*}$ is the isomorphism induced by $f$.
\end{lemma}

We also need the following standard group-theoretic lemmas. The first provides a key condition for identifying whether a non-Abelian simple group belongs to the socle. Namely, if $S_{1} \cong S_{2}$ are non-Abelian simple groups where $S_{1}$ is in the socle and $S_{2}$ is not in the socle, then the normal closures of $S_{1}$ and $S_{2}$ are non-isomorphic. In particular, the normal closure of $S_{1}$ is a direct product of non-Abelian simple groups, while the normal closure of $S_{2}$ is not a direct product of non-Abelian simple groups. We will apply this condition later when $S_{1}$ is a simple direct factor of $\Soc(G)$; in which case, the normal closure of $S_{1}$ is of the form $S_{1}^{k}$. 

\begin{lemma}[c.f. {\cite[Lemma~6.5]{GrochowLevetWL}}] \label{LemmaSocle}
Let $G$ be a finite semisimple group. A subgroup $S \leq G$ is contained in $\Soc(G)$ if and only if the normal closure of $S$ is a direct product of nonabelian simple groups.
\end{lemma}

\begin{lemma}[c.f. {\cite[Lemma~6.6]{GrochowLevetWL}}] \label{LemmaDirectProdSimple}
Let $S_1, \dotsc, S_k \leq G$ be nonabelian simple subgroups such that for all distinct $i,j \in [k]$ we have $[S_i, S_j] = 1$. Then $\langle S_1, \dotsc, S_k \rangle = S_1 S_2 \dotsb S_k = S_1 \times \dotsb \times S_k$.
\end{lemma}

\subsection{Main Results}
We show that the second Ehrenfeucht--Fra\"iss\'e game in Hella's hierarchy can identify both $\Soc(G)$ and the conjugation action when $G$ is semisimple. We first show that this pebble game can identify whether a group is semisimple. Namely, if $G$ is semisimple and $H$ is not semisimple, then Spoiler can distinguish $G$ from $H$. 

\begin{proposition} \label{IdentifySemisimple}
Let $G$ be a semisimple group of order $n$, and let $H$ be an arbitrary group of order $n$. If $H$ is not semisimple, then Spoiler can win in the $(4,2)$-WL$_{II}^{2}$ game.
\end{proposition}

\begin{proof}
To appear in the full version.
\end{proof}

We now apply Lemma \ref{LemmaSocle} to show that Duplicator must map the direct factors of $\Soc(G)$ to isomorphic direct factors of $\Soc(H)$. 

\begin{lemma} \label{LemmaProdSimple}
Let $G,H$ be finite groups of order $n$. Let $\Fac(\Soc(G))$ denote the set of simple direct factors of $\Soc(G)$. Let $S \in \Fac(\Soc(G))$ be a non-Abelian simple group, with $S = \langle x, y \rangle$. If Duplicator selects a bijection $f \colon G \to H$ such that:
\begin{enumerate}[label=(\alph*)]
\item $S \not \cong \langle f(x), f(y) \rangle$, then Spoiler can win in the $(2,1)$-WL$_{II}^{2}$ game; or 
\item $f(S) \neq \langle f(x), f(y) \rangle$, then Spoiler can win in the $(4,2)$-WL$_{II}^{2}$ pebble game.
\end{enumerate}
\end{lemma}

\noindent \\ Note that the lemma does not require $f|_{S}\colon S \to f(S)$ to actually be an isomorphism, only that $S$ and $f(S)$ are isomorphic.

\begin{proof}
To appear in the full version.
\end{proof}

\begin{proposition} \label{PropSocleSemisimple}
Let $G$ be a semisimple group of order $n$, and let $H$ be an arbitrary group of order $n$. Let $f : G \to H$ be the bijection Duplicator selects. If there exists $S \in \Fac(\Soc(G))$ such that $f(S) \notin \Fac(\Soc(H))$ or $f(S) \not\cong S$, then Spoiler can win in the $(4,2)$-WL$_{II}^{2}$ pebble game.
\end{proposition}

\begin{proof}
To appear in the full version.
\end{proof}

\begin{lemma} \label{LemmaSimpleOverlap}
Let $G,H$ be groups of order $n$, let $S$ be a nonabelian simple group in $\Fac(\Soc(G))$. Let $f,f'\colon G \to H$ be two bijections selected by Duplicator at two different rounds. If $f(S) \cap f'(S) \neq 1$, then $f(S) = f'(S)$, or Spoiler can win in the $(4,2)$-WL$_{II}^{2}$ pebble game.
\end{lemma}

\begin{proof}
By \Prop{PropSocleSemisimple}, both $f(S)$ and $f'(S)$ must be simple normal subgroups of $\Soc(H)$ (or Spoiler wins with $4$ pebbles and $2$ rounds). Since they intersect nontrivially, but distinct simple normal subgroups of $\Soc(H)$ intersect trivially, the two must be equal.
\end{proof}

We next introduce the notion of weight.

\begin{definition}
Let $\Soc(G) = S_1 \times \dotsb \times S_k$ where each $S_i$ is a simple normal subgroup of $\Soc(G)$. For any $s \in \Soc(G)$, write $s = s_1 s_2 \dotsb s_k$ where each $s_i \in S_i$, and define the \emph{weight} of $s$, denote $\wt(s)$, as the number of $i$'s such that $s_i \neq 1$.
\end{definition}

Note that the definition of weight is well-defined since the $S_i$ are the unique subsets of $\Soc(G)$ that are simple normal subgroup of $\Soc(G)$, so the decomposition $s = s_1 s_2 \dotsc s_k$ is unique up to the order of the factors. (This is essentially a particular instance of the ``rank lemma'' from \cite{GrochowLevetWL}, which intuitively states that WL detects in $O(\log n)$ rounds the set of elements for a given subgroup provided that it also identifies the generators. As we are now in the setting of $2$-ary WL we give the full proof, which also has tighter bounds on the number of rounds.)

\begin{lemma}[Weight Lemma] \label{LemmaSemisimpleWeight}
Let $G, H$ be semisimple groups of order $n$. If Duplicator selects a bijection $f\colon G \to H$ that does not map $\Soc(G)$ bijectively to $\Soc(H)$, or does not preserve the weight of every element in $\Soc(G)$, then Spoiler can win in the $(4,3)$-WL$_{II}^{2}$ game.
\end{lemma}

\begin{proof}
To appear in the full version.
\end{proof}

\begin{lemma} \label{SocleDirectProductStronger}
Let $G$ and $H$ be semisimple groups with isomorphic socles. Let $S_{1}, S_{2} \in \Fac(\Soc(G))$ be distinct. Let $f : G \to H$ be the bijection that Duplicator selects. If there exist $x_i \in S_i$ such that $f(x_1 x_2) \neq f(x_1) f(x_2)$, then Spoiler can win in the $(4,3)$-WL$_{II}^{2}$ pebble game. 
\end{lemma}

\begin{proof}
By \Lem{LemmaSemisimpleWeight}, we may assume that $\wt(s) = \wt(f(s))$ for all $s \in \Soc(G)$; otherwise, Spoiler wins with at most $4$ pebbles and $3$ rounds. As $f(x_{1}x_{2})$ has weight $2$, $f(x_{1}x_{2})$ belongs to the direct product of two simple factors in $\Fac(\Soc(H))$, so it can be written $f(x_1 x_2) = y_1 y_2$ with each $y_i$ in distinct simple factors in $\Fac(\Soc(H))$. Without loss of generality suppose that $y_1 \neq f(x_1)$. Spoiler pebbles $(x_1, x_1 x_2) \mapsto (f(x_1), f(x_1 x_2))$. Now $\wt(x_1^{-1} \cdot x_{1}x_{2})  = 1$, while $\wt(f(x_1)^{-1} \cdot f(x_{1}x_{2})) \geq 2$. (Note that we cannot quite yet directly apply \Lem{LemmaSemisimpleWeight}, because we have not yet identified a single element $x$ such that $\wt(x) \neq \wt(f(x))$.)

On the next round, Duplicator selects another bijection $f'$. Spoiler now pebbles $x_2 \mapsto f'(x_2)$. Because $\wt(x_1^{-1} \cdot x_1 x_2) = 1$ but $\wt(f(x_1)^{-1} f(x_1 x_2)) \geq 2$, and $f'$ preserves weight by \Lem{LemmaSemisimpleWeight}, we have $f'(x_2) \neq f'(x_1)^{-1} f'(x_1 x_2)$. Thus, the pebbled map $(x_1, x_2, x_1 x_2) \mapsto (f'(x_1), f(x_2), f(x_1 x_2))$ does not extend to an isomorphism, and so Spoiler wins with $3$ pebbles and $2$ rounds. 
\end{proof}

Recall that if $G$ is semisimple, then $G \leq \Aut(\Soc(G))$. Now each minimal normal subgroup $N \trianglelefteq G$ is of the form $N = S^{k}$, where $S$ is a non-Abelian simple group. So $\Aut(N) = \Aut(S) \wr \text{Sym}(k)$. In particular, 
\[
G \leq \prod_{ \substack{ N \trianglelefteq G \\ N \text{ is minimal normal}} } \Aut(N).
\]

So if $g \in G$, then the conjugation action of $g$ on $\Soc(G)$ acts by (i) automorphism on each simple direct factor of $\Soc(G)$, and (ii) by permuting the direct factors of $\Soc(G)$. Provided generators of the direct factors of the socle are pebbled, Spoiler can detect inconsistencies of the automorphism action. However, doing so directly would be too expensive as there could be $\Theta(\log|G|)$ generators, so we employ a more subtle approach with a similar outcome. By \Lem{LemmaSemisimpleWeight}, Duplicator must select bijections $f : G \to H$ that preserve weight. That is, if $s \in \Soc(G)$, then $\wt(s) = \wt(f(s))$. We use \Lem{LemmaSemisimpleWeight} in tandem with the fact that the direct factors of the socle commute to effectively pebble the set of all the generators at once. Namely, suppose that $\Fac(\Soc(G)) = \{ S_{1}, \ldots, S_{k}\}$, where $S_{i} = \langle x_{i}, y_{i} \rangle$. Let $x := x_{1} \cdots x_{k}$ and $y := y_{1} \cdots y_{k}$. We will show that it suffices for Spoiler to pebble $(x, y)$ rather than individually pebbling generators for each $S_{i}$ (this will still allow the factors to be permuted, but that is all).

\begin{lemma} \label{SemisimpleFactors}
Let $G$ and $H$ be semisimple groups with isomorphic socles, and write $\Fac(\Soc(G)) = \{S_1, \dotsc, S_m\}$, with $S_i = \langle x_i, y_i \rangle$. Let $f : G \to H$ be the bijection that Duplicator selects, and suppose that (i) for all $i$, $f(S_i) \cong S_i$ (though $f|_{S_i}$ need not be an isomorphism) and $f(S_i) \in \Fac(\Soc(H))$, (ii) for every $s \in \Soc(G)$, $\wt(s) = \wt(f(s))$, and
(iii) for all $i$, $f(S_i) = \langle f(x), f(y) \rangle$. 

Now suppose that Spoiler pebbles $(x_{1} \cdots x_{m}, y_{1} \cdots y_{m}) \mapsto (f(x_{1} \cdots x_{m}), f(y_{1} \cdots y_{m}))$. As $f$ preserves weight, we may write $f(x_{1} \cdots x_{m}) = h_{1} \cdots h_{m}$ and $f(y_{1} \cdots y_{m}) = z_{1} \cdots z_{m}$ with $h_i, z_i \in f(S_i)$ for all $i$.

Let $f' : G \to H$ be the bijection that Duplicator selects at any subsequent round in which the pebble used above has not moved. 
If any of the following hold, then Spoiler can win in the WL$_{II}^{2}$ pebble game with $5$ additional pebbles and $5$ additional rounds:
\begin{enumerate}[label=(\alph*)]
\item $f'$ does not satisfy conditions (i)--(iii),
\item there exists an $i \in [m]$ such that $f'(x_i) \notin \{h_1, \dotsc, h_m\}$ or $f'(y_i) \notin \{z_1, \dotsc, z_m\}$
\item $f'|_{S_{i}}$ is not an isomorphism
\item there exists $g \in G$ and $i \in [m]$ such that $gS_{i}g^{-1} = S_{i}$ and for some $x \in S_{i}$, the following holds: $f'(gxg^{-1}) \neq f'(g)f'(x)f'(g)^{-1}$.
\end{enumerate}
\end{lemma}

\begin{proof}
To appear in the full version.
\end{proof}

\Lem{SemisimpleFactors} provides enough to establish that Spoiler can force Duplicator to select at each round a bijection that restricts to an isomorphism on the socles.

\begin{proposition} \label{SemisimpleSocleIso}
(Same assumptions as \Lem{SemisimpleFactors}.)  Let $G$ and $H$ be semisimple groups with isomorphic socles, with $\Fac(\Soc(G)) = \{S_1, \dotsc, S_m\}$, with $S_i = \langle x_i, y_i \rangle$. Let $f_0 : G \to H$ be the bijection that Duplicator selects, and suppose that (i) for all $i$, $f_0(S_i) \cong S_i$ (though $f_0|_{S_i}$ need not be an isomorphism) and $f_0(S_i) \in \Fac(\Soc(H))$, (ii) for every $s \in \Soc(G)$, $\wt(s) = \wt(f_0(s))$, and
(iii) for all $i$, $f_0(S_i) = \langle f_0(x), f_0(y) \rangle$. Now suppose that Spoiler pebbles $(x_{1} \cdots x_{m}, y_{1} \cdots y_{m}) \mapsto (f_0(x_{1} \cdots x_{m}), f_0(y_{1} \cdots y_{m}))$. 

Let $f' \colon G \to H$ be the bijection that Duplicator selects at any subsequent round in which the pebbles used above have not moved. Then $f'|_{\Soc(G)}\colon \Soc(G) \to \Soc(H)$ must be an isomorphism, or Spoiler can win in $4$ more rounds using at most $6$ more pebbles (for a total of $7$ pebbles and $5$ rounds) in the WL$_{II}^{2}$ pebble game.
\end{proposition}

\begin{proof}
To appear in the full version.
\end{proof}

\begin{remark}
Brachter \& Schweitzer \cite[Lemma~5.22]{BrachterSchweitzerWLLibrary} previously showed that (1-ary) Weisfeiler--Leman can decide whether two groups have isomorphic socles. However, their results did not solve the search problem; that is, they did not show Duplicator must select bijections that restrict to an isomorphism on the socle even in the case for semisimple groups. This contrasts with \Lem{SemisimpleSocleIso}, where we show that 2-ary WL effectively solves the search problem. This is an important ingredient in our proof that the $(7, O(1))$-WL$_{II}^{2}$ pebble game solves isomorphism for semisimple groups.
\end{remark}

We obtain as a corollary of \Lem{SemisimpleFactors} and \Lem{SemisimpleSocleIso} that if $G$ and $H$ are semisimple, then Duplicator must select bijections that restrict to isomorphisms of $\pker(G)$ and $\pker(H)$.

\begin{corollary} \label{CorPKerIso}
Let $G$ and $H$ be semisimple groups of order $n$. Let $\Fac(\Soc(G)) := \{ S_{1}, \ldots, S_{m}\}$, and suppose that $S_{i} = \langle x_{i}, y_{i} \rangle$. Let $x := x_{1} \cdots x_{m}$ and $y := y_{1} \cdots y_{m}$. and  Let $f : G \to H$ be the bijection that Duplicator selects. Spoiler begins by pebbling $(x, y) \mapsto (f(x), f(y))$. Let $f' : G \to H$ be the bijection that Duplicator selects at the next round. If $f'|_{\pker(G)} : \pker(G) \to \pker(H)$ is not an isomorphism, then Spoiler can win with $5$ additional pebbles and $5$ additional rounds in the WL$_{II}^{2}$ pebble game.
\end{corollary}

\begin{proof}
To appear in the full version.
\end{proof}

We now show that if $G$ and $H$ are not permutationally equivalent, then Spoiler can win.

\begin{lemma} \label{PermutationalIso}
(Same assumptions as \Lem{SemisimpleFactors}.)  Let $G$ and $H$ be semisimple groups with isomorphic socles, with $\Fac(\Soc(G)) = \{S_1, \dotsc, S_m\}$, with $S_i = \langle x_i, y_i \rangle$. Let $f_0 : G \to H$ be the bijection that Duplicator selects, and suppose that (i) for all $i$, $f_0(S_i) \cong S_i$ (though $f_0|_{S_i}$ need not be an isomorphism) and $f_0(S_i) \in \Fac(\Soc(H))$, (ii) for every $s \in \Soc(G)$, $\wt(s) = \wt(f_0(s))$, and
(iii) for all $i$, $f_0(S_i) = \langle f_0(x), f_0(y) \rangle$. Now suppose that Spoiler pebbles $(x_{1} \cdots x_{m}, y_{1} \cdots y_{m}) \mapsto (f_0(x_{1} \cdots x_{m}), f_0(y_{1} \cdots y_{m}))$. 

Let $f' : G \to H$ be the bijection that Duplicator selects at the next round. Suppose that there exist $g \in G$ and $i \in [m]$ such that $f'(gS_{i}g^{-1}) = f'(S_{j})$, but $f'(g)f'(S_{i})f'(g)^{-1} = f'(S_{k})$ for some $k \neq j$. Then Spoiler can win with $4$ additional pebbles and $4$ additional rounds in the WL$_{II}^{2}$ pebble game.
\end{lemma}

\begin{proof}
To appear in the full version.
\end{proof}

\begin{theorem} \label{SemisimpleIsomorphism}
Let $G$ be a semisimple group and $H$ an arbitrary group of order $n$, not isomorphic to $G$. Then Spoiler has a winning strategy in the $(9, O(1))$-WL$_{II}^{2}$ pebble game.
\end{theorem}

\begin{proof}
If $H$ is not semisimple, then by \Prop{IdentifySemisimple}, Spoiler wins with $4$ pebbles and $2$ rounds. So we now suppose $H$ is semisimple.

Let $\Fac(\Soc(G)) = \{S_1, \dotsc, S_k\}$, and let $x_i, y_i$ be generators of $S_i$ for each $i$. Let $f$ be the bijection chosen by Duplicator. Spoiler pebbles $(x_1 x_2 \dotsb x_k, y_1 y_2, \dotsc, y_k) \mapsto (f(x_1 \dotsb x_k), f(y_1 \dotsb y_k))$. On subsequent rounds, we thus have satisfied the hypotheses of \Lem{SemisimpleFactors} and \Prop{SemisimpleSocleIso}. Spoiler will never move this pebble, and thus all subsequent bijections chosen by Duplicator must restrict to isomorphisms on the socle (or Spoiler wins with at most $7$ pebbles and $O(1)$ rounds).

Recall from \Lem{CharacterizeSemisimple} that $G \cong H$ iff there is an isomorphism $\mu\colon \Soc(G) \to \Soc(H)$ that induces a permutational isomorphism $\mu^*\colon G^* \to H^*$. Thus, since $G \not\cong H$, there must be some $g \in G$ and $s \in \Soc(G)$ such that $f(gsg^{-1}) \neq f(g)f(s)f(g)^{-1}$. Write $s = s_1 \dotsb s_k$ with each $s_i \in S_i$ (not necessarily nontrivial). We claim that there exists some $i$ such that $f(gs_i g^{-1}) \neq f(g) f(s_i) f(g)^{-1}$. For suppose not, then we have
\begin{eqnarray*}
f(gsg^{-1}) & = & f(gs_1 g^{-1} gs_2 g^{-1} \dotsb g s_k g^{-1}) \\
& = &  f(gs_1g^{-1})f(gs_2 g^{-1}) \dotsb f(gs_k g^{-1}) \\
& = & f(g) f(s_1) f(g)^{-1} f(g) f(s_2) f(g)^{-1} \dotsb f(g) f(s_k) f(g)^{-1} \\
& = & f(g) f(s_1 \dotsb s_k) f(g)^{-1} = f(g) f(s) f(g)^{-1},
\end{eqnarray*}
a contradiction. For simplicity of notation, without loss of generality we may assume $i=1$, so we now have $f(gs_1 g)^{-1} \neq f(g) f(s_1) f(g)^{-1}$.

We break the argument into cases:

\begin{enumerate}
\item If $gs_1 g^{-1} \in S_{1}$, then we have $gS_1 g^{-1} = S_1$ (any two distinct simple normal factors of the socle intersect trivially), we have by \Lem{SemisimpleFactors}(d) that Spoiler can win with at most $5$ additional pebbles (for a total of $7$ pebbles) and $5$ additional rounds (for a total of $6$ rounds). 

\item If $gs_1 g^{-1} \in S_j$ for $j \neq 1$ and $f(g)f(s_{1})f(g)^{-1} \notin f(S_j)$, we have by \Lem{PermutationalIso} that Spoiler can win with at most $4$ additional pebbles (for a total of $6$ pebbles) and $4$ additional rounds (for a total of $5$ rounds).

\item Suppose now that $gs_1 g^{-1} \in S_j$ for some $j \neq 1$ and $f(g)f(s_{1})f(g)^{-1} \in f(S_j)$. Spoiler begins by pebbling $(g, gs_{1}g^{-1}) \mapsto (f(g), f(gs_{1}g^{-1}))$. Let $f' : G \to H$ be the bijection that Duplicator selects at the next round. As $gs_{1}g^{-1} \in S_{j}$ is pebbled, we have that $f'(S_{j}) = f(S_{j})$ by Lem.~\ref{LemmaSimpleOverlap} (or Spoiler wins with $4$ additional pebbles and $2$ additional rounds). Now by assumption, $gS_{1}g^{-1} = S_{j}$ and $f(g)f(S_{1})f(g)^{-1} = f(S_{j})$. So as $g \mapsto f(g)$ is pebbled, we claim that we may assume $f'(S_{1}) = f(S_{1})$. For suppose not; then we have $g^{-1} S_j g = S_1$ but $f'(g)^{-1} f'(S_j) f'(g) = f(g)^{-1} f(S_j) f(g) = f(S_1) \neq f'(S_1)$. But then Spoiler can with win with $4$ additional pebbles (for a total of $8$ pebbles) and $4$ additional rounds (for a total of $7$ rounds) by \Lem{PermutationalIso}. Thus we have $f'(S_1) = f(S_1)$.

In particular, we have that $f'(x_{1}) = f(x_{1})$ and $f'(y_{1}) = f(y_{1})$, by the same argument as in the proof of \Lem{SemisimpleFactors}(c). As $S_{1} = \langle x_{1}, y_{1} \rangle$, we have that $f'(s_{1}) = f(s_{1})$, since they are both isomorphisms on the socle by \Prop{SemisimpleSocleIso}. Spoiler now pebbles $(x_{1}, y_{1}) \mapsto (f'(x_{1}), f'(y_{1}))$. As the pebbled map $(g, x_{1}, y_{1}, gs_{1}g^{-1}) \mapsto (f(g),  f'(x_{1}), f'(y_{1}), f'(gs_{1}g^{-1}))$ does not extend to an isomorphism, Spoiler wins. In this case, Spoiler used at most $8$ pebbles and $7$ rounds.
\end{enumerate}

\noindent Note that the ninth pebble is the one we pick up prior to checking the winning condition.
\end{proof}

\section{Conclusion}

We exhibited a novel Weisfeiler--Leman algorithm that provides an algorithmic characterization of the second Ehrenfeucht--Fra\"iss\'e game in Hella's \cite{Hella1989, Hella1993} hierarchy. We also showed that this Ehrenfeucht--Fra\"iss\'e game can identify groups without Abelian normal subgroups using $O(1)$ pebbles and $O(1)$ rounds. In particular, within the first few rounds, Spoiler can force Duplicator to select an isomorphism at each subsequent round. This effectively solves the search problem in the pebble game characterization. 

Our work leaves several directions for further research.

\begin{question}
Can the constant-dimensional $2$-ary Wesifeiler--Leman algorithm be implemented in time $n^{o(\log n)}$?
\end{question}

\begin{question}
What is the (1-ary) Weisfeiler--Leman dimension of groups without Abelian normal subgroups?
\end{question}

\begin{question}
Show that the second Ehrenfeucht--Fra\"iss\'e game in Hella's hierarchy can identify coprime extensions of the form $H \ltimes N$ with both $H,N$ Abelian (the analogue of \cite{QST11}). More generally, an analogue of Babai--Qiao \cite{BQ} would be to show that when $|H|,|N|$ are coprime and $N$ is Abelian, that Spoiler can distinguish $H \ltimes N$ from any non-isomorphic group using a constant number of pebbles that is no more than that which is required to identify $H$ (or the maximum of that of $H$ and a constant independent of $N,H$). 
\end{question}

\begin{question}
Let $p > 2$ be prime, and let $G$ be a $p$-group with bounded genus. Show that in the second Ehrenfeucht--Fra\"iss\'e game in Hella's hierarchy, Spoiler has a winning strategy using a constant number of pebbles. This is a descriptive complexity analogue of  \cite{BMWGenus2, IvanyosQ19}. It would even be of interest to start with the case where $G$ has bounded genus over a field extension $K/\mathbb{F}_{p}$ of bounded degree.
\end{question}

In the setting of groups, Hella's hierarchy collapses to some $q \leq 3$, since 3-ary WL can identify all ternary relational structures, including groups. It remains open to determine whether this hierarchy collapses further to either $q = 1$ or $q = 2$. Even if it does not collapse, it would also be of interest to determine whether the $1$-ary and $2$-ary games are equivalent. Algorithmically, this is equivalent to determining whether $1$-ary and $2$-ary WL are have the same distinguishing power.

\begin{question}
Does there exist an infinite family of non-isomorphic pairs of groups $\{ (G_{n}, H_{n})\}$ for which Spoiler requires $\omega(1)$ pebbles to distinguish $G_{n}$ from $H_{n}$? We ask this question for the Ehrenfeucht--Fra\"iss\'e games at both the first and second levels of Hella's hierarchy. 
\end{question}

Recall that the game at the first level of Hella's hierarchy is equivalent to Weisfeiler--Leman \cite{CFI, Hella1989, Hella1993}, and so a lower bound against either of these games provides a lower bound against Weisfeiler--Leman. More generally, it would also be of interest to investigate Hella's hierarchy on higher arity structures. For a $q$-ary relational structure, the $q$-ary pebble game suffices to decide isomorphism. Are there interesting, natural classes of higher arity structures for which Hella's hierarchy collapses further to some level $q' < q$?

\section*{Acknowledgment}
The authors would like to thank Pascal Schweitzer and the anonymous referees for feedback on earlier versions of some of the results in this paper, as well as discussions more directly relevant to the paper. JAG would like to thank Martin Grohe for helpful correspondence about WL versus ``oblivious WL.'' JAG was partially supported by NSF award DMS-1750319 and NSF CAREER award CCF-2047756 and during this work. ML was partially supported by J. Grochow startup funds, NSF award CISE-2047756, and a Summer Research Fellowship through the Department of Computer Science at the University of Colorado Boulder. 

\bibliographystyle{eptcs}
\bibliography{references}

\begin{thebibliography}{10}
\providecommand{\bibitemdeclare}[2]{}
\providecommand{\surnamestart}{}
\providecommand{\surnameend}{}
\providecommand{\urlprefix}{Available at }
\providecommand{\url}[1]{\texttt{#1}}
\providecommand{\href}[2]{\texttt{#2}}
\providecommand{\urlalt}[2]{\href{#1}{#2}}
\providecommand{\doi}[1]{doi:\urlalt{https://doi.org/#1}{#1}}
\providecommand{\eprint}[1]{arXiv:\urlalt{https://arxiv.org/abs/#1}{#1}}
\providecommand{\bibinfo}[2]{#2}

\bibitemdeclare{article}{ajtaifagin1990}
\bibitem{ajtaifagin1990}
\bibinfo{author}{Miklos \surnamestart Ajtai\surnameend} \&
  \bibinfo{author}{Ronald \surnamestart Fagin\surnameend}
  (\bibinfo{year}{1990}): \emph{\bibinfo{title}{Reachability is harder for
  directed than for undirected finite graphs}}.
\newblock {\slshape \bibinfo{journal}{Journal of Symbolic Logic}}
  \bibinfo{volume}{55}(\bibinfo{number}{1}), p. \bibinfo{pages}{113–150},
  \doi{10.2307/2274958}.

\bibitemdeclare{article}{ARORA199797}
\bibitem{ARORA199797}
\bibinfo{author}{Sanjeev \surnamestart Arora\surnameend} \&
  \bibinfo{author}{Ronald \surnamestart Fagin\surnameend}
  (\bibinfo{year}{1997}): \emph{\bibinfo{title}{On winning strategies in
  {Ehrenfeucht--Fra\"{i}ss\'{e}} games}}.
\newblock {\slshape \bibinfo{journal}{Theoretical Computer Science}}
  \bibinfo{volume}{174}(\bibinfo{number}{1}), pp. \bibinfo{pages}{97--121},
  \doi{10.1016/S0304-3975(96)00015-1}.

\bibitemdeclare{article}{ArvindKurur}
\bibitem{ArvindKurur}
\bibinfo{author}{V.~\surnamestart Arvind\surnameend} \&
  \bibinfo{author}{Piyush~P. \surnamestart Kurur\surnameend}
  (\bibinfo{year}{2006}): \emph{\bibinfo{title}{Graph Isomorphism is in
  {SPP}}}.
\newblock {\slshape \bibinfo{journal}{Information and Computation}}
  \bibinfo{volume}{204}(\bibinfo{number}{5}), pp. \bibinfo{pages}{835--852},
  \doi{10.1016/j.ic.2006.02.002}.

\bibitemdeclare{inproceedings}{BabaiGraphIso}
\bibitem{BabaiGraphIso}
\bibinfo{author}{L\'{a}szl\'{o} \surnamestart Babai\surnameend}
  (\bibinfo{year}{2016}): \emph{\bibinfo{title}{Graph isomorphism in
  quasipolynomial time [extended abstract]}}.
\newblock In: {\slshape \bibinfo{booktitle}{S{TOC}'16---{P}roceedings of the
  48th {A}nnual {ACM} {SIGACT} {S}ymposium on {T}heory of {C}omputing}},
  \bibinfo{publisher}{ACM, New York}, pp. \bibinfo{pages}{684--697},
  \doi{10.1145/2897518.2897542}.
\newblock \bibinfo{note}{Preprint of full version at
  \arXiv{1512.03547v2}{[cs.DS]}}.

\bibitemdeclare{inproceedings}{Babai1999GroupsSA}
\bibitem{Babai1999GroupsSA}
\bibinfo{author}{L{\'a}szl{\'o} \surnamestart Babai\surnameend} \&
  \bibinfo{author}{Robert \surnamestart Beals\surnameend}
  (\bibinfo{year}{1999}): \emph{\bibinfo{title}{A polynomial-time theory of
  black box groups {I}}}.
\newblock In: {\slshape \bibinfo{booktitle}{Groups {S}t {A}ndrews 1997 in
  {B}ath, {I}}}, \doi{10.1017/CBO9781107360228.004}.

\bibitemdeclare{inproceedings}{BabaiBealsSeress}
\bibitem{BabaiBealsSeress}
\bibinfo{author}{L\'{a}szl\'{o} \surnamestart Babai\surnameend},
  \bibinfo{author}{Robert \surnamestart Beals\surnameend} \&
  \bibinfo{author}{\'{A}kos \surnamestart Seress\surnameend}
  (\bibinfo{year}{2009}): \emph{\bibinfo{title}{Polynomial-Time Theory of
  Matrix Groups}}.
\newblock In: {\slshape \bibinfo{booktitle}{Proceedings of the Forty-First
  Annual ACM Symposium on Theory of Computing}}, \bibinfo{series}{STOC '09},
  \bibinfo{publisher}{Association for Computing Machinery}, p.
  \bibinfo{pages}{55–64}, \doi{10.1145/1536414.1536425}.

\bibitemdeclare{inproceedings}{BCGQ}
\bibitem{BCGQ}
\bibinfo{author}{L{\'a}szl{\'o} \surnamestart Babai\surnameend},
  \bibinfo{author}{Paolo \surnamestart Codenotti\surnameend},
  \bibinfo{author}{Joshua~A. \surnamestart Grochow\surnameend} \&
  \bibinfo{author}{Youming \surnamestart Qiao\surnameend}
  (\bibinfo{year}{2011}): \emph{\bibinfo{title}{Code equivalence and group
  isomorphism}}.
\newblock In: {\slshape \bibinfo{booktitle}{Proceedings of the {Twenty-Second}
  {Annual} {ACM--SIAM} {Symposium} on {Discrete} {Algorithms} ({SODA11})}},
  \bibinfo{publisher}{SIAM}, \bibinfo{address}{Philadelphia, PA}, pp.
  \bibinfo{pages}{1395--1408}, \doi{10.1137/1.9781611973082.107}.

\bibitemdeclare{inproceedings}{BCQ}
\bibitem{BCQ}
\bibinfo{author}{L{\'a}szl{\'o} \surnamestart Babai\surnameend},
  \bibinfo{author}{Paolo \surnamestart Codenotti\surnameend} \&
  \bibinfo{author}{Youming \surnamestart Qiao\surnameend}
  (\bibinfo{year}{2012}): \emph{\bibinfo{title}{Polynomial-Time Isomorphism
  Test for Groups with No Abelian Normal Subgroups - (Extended Abstract)}}.
\newblock In: {\slshape \bibinfo{booktitle}{International Colloquium on
  Automata, Languages, and Programming (ICALP)}}, pp. \bibinfo{pages}{51--62},
  \doi{10.1007/978-3-642-31594-7_5}.

\bibitemdeclare{inproceedings}{BQ}
\bibitem{BQ}
\bibinfo{author}{L\'aszl\'o \surnamestart Babai\surnameend} \&
  \bibinfo{author}{Youming \surnamestart Qiao\surnameend}
  (\bibinfo{year}{2012}): \emph{\bibinfo{title}{Polynomial-time Isomorphism
  Test for Groups with {Abelian} {Sylow} Towers}}.
\newblock In: {\slshape \bibinfo{booktitle}{29th STACS}},
  \bibinfo{publisher}{Springer LNCS 6651}, pp. \bibinfo{pages}{453 -- 464},
  \doi{10.4230/LIPIcs.STACS.2012.453}.

\bibitemdeclare{article}{BE99}
\bibitem{BE99}
\bibinfo{author}{Hans~Ulrich \surnamestart Besche\surnameend} \&
  \bibinfo{author}{Bettina \surnamestart Eick\surnameend}
  (\bibinfo{year}{1999}): \emph{\bibinfo{title}{Construction of finite
  groups}}.
\newblock {\slshape \bibinfo{journal}{J. Symb. Comput.}}
  \bibinfo{volume}{27}(\bibinfo{number}{4}), pp. \bibinfo{pages}{387--404},
  \doi{10.1006/jsco.1998.0258}.

\bibitemdeclare{article}{BEO02}
\bibitem{BEO02}
\bibinfo{author}{Hans~Ulrich \surnamestart Besche\surnameend},
  \bibinfo{author}{Bettina \surnamestart Eick\surnameend} \&
  \bibinfo{author}{E.A. \surnamestart O'Brien\surnameend}
  (\bibinfo{year}{2002}): \emph{\bibinfo{title}{A Millennium Project:
  Constructing Small Groups}}.
\newblock {\slshape \bibinfo{journal}{Intern. J. Alg. and Comput}}
  \bibinfo{volume}{12}, pp. \bibinfo{pages}{623--644},
  \doi{10.1142/S0218196702001115}.

\bibitemdeclare{inproceedings}{WLGroups}
\bibitem{WLGroups}
\bibinfo{author}{Jendrik \surnamestart Brachter\surnameend} \&
  \bibinfo{author}{Pascal \surnamestart Schweitzer\surnameend}
  (\bibinfo{year}{2020}): \emph{\bibinfo{title}{On the {Weisfeiler--Leman}
  Dimension of Finite Groups}}.
\newblock In \bibinfo{editor}{Holger \surnamestart Hermanns\surnameend},
  \bibinfo{editor}{Lijun \surnamestart Zhang\surnameend},
  \bibinfo{editor}{Naoki \surnamestart Kobayashi\surnameend} \&
  \bibinfo{editor}{Dale \surnamestart Miller\surnameend}, editors: {\slshape
  \bibinfo{booktitle}{{LICS} '20: 35th Annual {ACM/IEEE} Symposium on Logic in
  Computer Science, Saarbr{\"{u}}cken, Germany, July 8-11, 2020}},
  \bibinfo{publisher}{{ACM}}, pp. \bibinfo{pages}{287--300},
  \doi{10.1145/3373718.3394786}.

\bibitemdeclare{misc}{BrachterSchweitzerWLLibrary}
\bibitem{BrachterSchweitzerWLLibrary}
\bibinfo{author}{Jendrik \surnamestart Brachter\surnameend} \&
  \bibinfo{author}{Pascal \surnamestart Schweitzer\surnameend}
  (\bibinfo{year}{2021}): \emph{\bibinfo{title}{A Systematic Study of
  Isomorphism Invariants of Finite Groups via the {Weisfeiler--Leman}
  Dimension}}.
\newblock \bibinfo{howpublished}{\arXiv{2111.11908}{[math.GR]}}.

\bibitemdeclare{misc}{BGLQW}
\bibitem{BGLQW}
\bibinfo{author}{Peter~A. \surnamestart Brooksbank\surnameend},
  \bibinfo{author}{Joshua~A. \surnamestart Grochow\surnameend},
  \bibinfo{author}{Yinan \surnamestart Li\surnameend}, \bibinfo{author}{Youming
  \surnamestart Qiao\surnameend} \& \bibinfo{author}{James~B. \surnamestart
  Wilson\surnameend} (\bibinfo{year}{2019}):
  \emph{\bibinfo{title}{Incorporating {Weisfeiler}--{Leman} into algorithms for
  group isomorphism}}.
\newblock \bibinfo{howpublished}{\arXiv{1905.02518}{[cs.CC]}}.

\bibitemdeclare{article}{BMWGenus2}
\bibitem{BMWGenus2}
\bibinfo{author}{Peter~A. \surnamestart Brooksbank\surnameend},
  \bibinfo{author}{Joshua \surnamestart Maglione\surnameend} \&
  \bibinfo{author}{James~B. \surnamestart Wilson\surnameend}
  (\bibinfo{year}{2017}): \emph{\bibinfo{title}{A fast isomorphism test for
  groups whose {Lie} algebra has genus 2}}.
\newblock {\slshape \bibinfo{journal}{Journal of Algebra}}
  \bibinfo{volume}{473}, pp. \bibinfo{pages}{545--590},
  \doi{10.1016/j.jalgebra.2016.12.007}.

\bibitemdeclare{inproceedings}{BuhrmanHomer}
\bibitem{BuhrmanHomer}
\bibinfo{author}{Harry \surnamestart Buhrman\surnameend} \&
  \bibinfo{author}{Steven \surnamestart Homer\surnameend}
  (\bibinfo{year}{1992}): \emph{\bibinfo{title}{Superpolynomial Circuits,
  Almost Sparse Oracles and the Exponential Hierarchy}}.
\newblock In \bibinfo{editor}{R.~K. \surnamestart Shyamasundar\surnameend},
  editor: {\slshape \bibinfo{booktitle}{Foundations of Software Technology and
  Theoretical Computer Science, 12th Conference, New Delhi, India, December
  18-20, 1992, Proceedings}}, {\slshape \bibinfo{series}{Lecture Notes in
  Computer Science}} \bibinfo{volume}{652}, \bibinfo{publisher}{Springer}, pp.
  \bibinfo{pages}{116--127}, \doi{10.1007/3-540-56287-7\_99}.

\bibitemdeclare{article}{CFI}
\bibitem{CFI}
\bibinfo{author}{Jin-Yi \surnamestart Cai\surnameend}, \bibinfo{author}{Martin
  \surnamestart F\"{u}rer\surnameend} \& \bibinfo{author}{Neil \surnamestart
  Immerman\surnameend} (\bibinfo{year}{1992}): \emph{\bibinfo{title}{An optimal
  lower bound on the number of variables for graph identification}}.
\newblock {\slshape \bibinfo{journal}{Combinatorica}}
  \bibinfo{volume}{12}(\bibinfo{number}{4}), pp. \bibinfo{pages}{389--410},
  \doi{10.1007/BF01305232}.
\newblock \bibinfo{note}{Originally appeared in SFCS '89}.

\bibitemdeclare{article}{CH03}
\bibitem{CH03}
\bibinfo{author}{John~J. \surnamestart Cannon\surnameend} \&
  \bibinfo{author}{Derek~F. \surnamestart Holt\surnameend}
  (\bibinfo{year}{2003}): \emph{\bibinfo{title}{Automorphism group computation
  and isomorphism testing in finite groups}}.
\newblock {\slshape \bibinfo{journal}{J. Symb. Comput.}} \bibinfo{volume}{35},
  pp. \bibinfo{pages}{241--267}, \doi{10.1016/S0747-7171(02)00133-5}.

\bibitemdeclare{article}{ChattopadhyayToranWagner}
\bibitem{ChattopadhyayToranWagner}
\bibinfo{author}{Arkadev \surnamestart Chattopadhyay\surnameend},
  \bibinfo{author}{Jacobo \surnamestart Tor\'{a}n\surnameend} \&
  \bibinfo{author}{Fabian \surnamestart Wagner\surnameend}
  (\bibinfo{year}{2013}): \emph{\bibinfo{title}{Graph isomorphism is not {$\rm
  AC^0$}-reducible to group isomorphism}}.
\newblock {\slshape \bibinfo{journal}{ACM Trans. Comput. Theory}}
  \bibinfo{volume}{5}(\bibinfo{number}{4}), pp. \bibinfo{pages}{Art. 13, 13},
  \doi{10.1145/2540088}.
\newblock \bibinfo{note}{Preliminary version appeared in FSTTCS '10; ECCC Tech.
  Report TR10-117}.

\bibitemdeclare{inproceedings}{DasSharma}
\bibitem{DasSharma}
\bibinfo{author}{Bireswar \surnamestart Das\surnameend} \&
  \bibinfo{author}{Shivdutt \surnamestart Sharma\surnameend}
  (\bibinfo{year}{2019}): \emph{\bibinfo{title}{Nearly Linear Time Isomorphism
  Algorithms for Some Nonabelian Group Classes}}.
\newblock In \bibinfo{editor}{Ren{\'e} \surnamestart van Bevern\surnameend} \&
  \bibinfo{editor}{Gregory \surnamestart Kucherov\surnameend}, editors:
  {\slshape \bibinfo{booktitle}{Computer Science -- Theory and Applications}},
  \bibinfo{publisher}{Springer International Publishing},
  \bibinfo{address}{Cham}, pp. \bibinfo{pages}{80--92},
  \doi{10.1007/s00224-020-10010-z}.

\bibitemdeclare{inproceedings}{DawarHolm}
\bibitem{DawarHolm}
\bibinfo{author}{Anuj \surnamestart Dawar\surnameend} \&
  \bibinfo{author}{Bjarki \surnamestart Holm\surnameend}
  (\bibinfo{year}{2012}): \emph{\bibinfo{title}{Pebble Games with Algebraic
  Rules}}.
\newblock In \bibinfo{editor}{Artur \surnamestart Czumaj\surnameend},
  \bibinfo{editor}{Kurt \surnamestart Mehlhorn\surnameend},
  \bibinfo{editor}{Andrew \surnamestart Pitts\surnameend} \&
  \bibinfo{editor}{Roger \surnamestart Wattenhofer\surnameend}, editors:
  {\slshape \bibinfo{booktitle}{Automata, Languages, and Programming}},
  \bibinfo{publisher}{Springer Berlin Heidelberg}, \bibinfo{address}{Berlin,
  Heidelberg}, pp. \bibinfo{pages}{251--262},
  \doi{10.1007/978-3-642-31585-5_25}.

\bibitemdeclare{article}{DawarVagnozzi}
\bibitem{DawarVagnozzi}
\bibinfo{author}{Anuj \surnamestart Dawar\surnameend} \& \bibinfo{author}{Danny
  \surnamestart Vagnozzi\surnameend} (\bibinfo{year}{2020}):
  \emph{\bibinfo{title}{Generalizations of k-dimensional Weisfeiler–Leman
  stabilization}}.
\newblock {\slshape \bibinfo{journal}{Moscow Journal of Combinatorics and
  Number Theory}} \bibinfo{volume}{9}, pp. \bibinfo{pages}{229--252},
  \doi{10.2140/moscow.2020.9.229}.

\bibitemdeclare{misc}{DietrichWilson}
\bibitem{DietrichWilson}
\bibinfo{author}{Heiko \surnamestart Dietrich\surnameend} \&
  \bibinfo{author}{James~B. \surnamestart Wilson\surnameend}
  (\bibinfo{year}{2022}): \emph{\bibinfo{title}{Polynomial-time isomorphism
  testing for groups of most finite orders}},
  \doi{10.1109/FOCS52979.2021.00053}.

\bibitemdeclare{article}{Ehrenfeucht}
\bibitem{Ehrenfeucht}
\bibinfo{author}{A.~\surnamestart Ehrenfeucht\surnameend}
  (\bibinfo{year}{1960/61}): \emph{\bibinfo{title}{An application of games to
  the completeness problem for formalized theories}}.
\newblock {\slshape \bibinfo{journal}{Fund. Math.}} \bibinfo{volume}{49}, pp.
  \bibinfo{pages}{129--141}, \doi{10.4064/fm-49-2-129-141}.

\bibitemdeclare{article}{ELGO02}
\bibitem{ELGO02}
\bibinfo{author}{Bettina \surnamestart Eick\surnameend}, \bibinfo{author}{C.~R.
  \surnamestart Leedham-Green\surnameend} \& \bibinfo{author}{E.~A.
  \surnamestart O'Brien\surnameend} (\bibinfo{year}{2002}):
  \emph{\bibinfo{title}{Constructing automorphism groups of {$p$}-groups}}.
\newblock {\slshape \bibinfo{journal}{Comm. Algebra}}
  \bibinfo{volume}{30}(\bibinfo{number}{5}), pp. \bibinfo{pages}{2271--2295},
  \doi{10.1081/AGB-120003468}.

\bibitemdeclare{article}{GroheFlum}
\bibitem{GroheFlum}
\bibinfo{author}{Jörg \surnamestart Flum\surnameend} \&
  \bibinfo{author}{Martin \surnamestart Grohe\surnameend}
  (\bibinfo{year}{2000}): \emph{\bibinfo{title}{On Fixed-Point Logic with
  Counting}}.
\newblock {\slshape \bibinfo{journal}{The Journal of Symbolic Logic}}
  \bibinfo{volume}{65}(\bibinfo{number}{2}), pp. \bibinfo{pages}{777--787},
  \doi{10.2307/2586569}.

\bibitemdeclare{article}{Fraisse}
\bibitem{Fraisse}
\bibinfo{author}{Roland \surnamestart Fra\"{\i}ss\'{e}\surnameend}
  (\bibinfo{year}{1954}): \emph{\bibinfo{title}{Sur quelques classifications
  des syst\`emes de relations}}.
\newblock {\slshape \bibinfo{journal}{Publ. Sci. Univ. Alger. S\'{e}r. A}}
  \bibinfo{volume}{1}, pp. \bibinfo{pages}{35--182 (1955)}.

\bibitemdeclare{article}{DescriptiveComplexityAbelianGroups}
\bibitem{DescriptiveComplexityAbelianGroups}
\bibinfo{author}{Walid \surnamestart Gomaa\surnameend} (\bibinfo{year}{2010}):
  \emph{\bibinfo{title}{Descriptive Complexity of Finite Abelian Groups.}}
\newblock {\slshape \bibinfo{journal}{IJAC}} \bibinfo{volume}{20}, pp.
  \bibinfo{pages}{1087--1116}, \doi{10.1142/S0218196710006047}.

\bibitemdeclare{misc}{grochow2022descriptive}
\bibitem{grochow2022descriptive}
\bibinfo{author}{Joshua~A. \surnamestart Grochow\surnameend} \&
  \bibinfo{author}{Michael \surnamestart Levet\surnameend}
  (\bibinfo{year}{2022}): \emph{\bibinfo{title}{On the Descriptive Complexity
  of Groups without Abelian Normal Subgroups}}.
\newblock \eprint{2209.13725}.

\bibitemdeclare{misc}{GrochowLevetWL}
\bibitem{GrochowLevetWL}
\bibinfo{author}{Joshua~A. \surnamestart Grochow\surnameend} \&
  \bibinfo{author}{Michael \surnamestart Levet\surnameend}
  (\bibinfo{year}{2022}): \emph{\bibinfo{title}{On the parallel complexity of
  Group Isomorphism and canonization via {Weisfeiler--Leman}}}.
\newblock \bibinfo{howpublished}{\arXiv{2112.11487}{[cs.DS]}}.

\bibitemdeclare{inproceedings}{GQ15}
\bibitem{GQ15}
\bibinfo{author}{Joshua~A. \surnamestart Grochow\surnameend} \&
  \bibinfo{author}{Youming \surnamestart Qiao\surnameend}
  (\bibinfo{year}{2015}): \emph{\bibinfo{title}{Polynomial-Time Isomorphism
  Test of Groups that are Tame Extensions - (Extended Abstract)}}.
\newblock In: {\slshape \bibinfo{booktitle}{Algorithms and Computation - 26th
  International Symposium, {ISAAC} 2015, Nagoya, Japan, December 9-11, 2015,
  Proceedings}}, pp. \bibinfo{pages}{578--589},
  \doi{10.1007/978-3-662-48971-0_49}.

\bibitemdeclare{book}{GroheBook}
\bibitem{GroheBook}
\bibinfo{author}{Martin \surnamestart Grohe\surnameend} (\bibinfo{year}{2017}):
  \emph{\bibinfo{title}{Descriptive complexity, canonisation, and definable
  graph structure theory}}.
\newblock {\slshape \bibinfo{series}{Lecture Notes in
  Logic}}~\bibinfo{volume}{47}, \bibinfo{publisher}{Association for Symbolic
  Logic, Ithaca, NY; Cambridge University Press, Cambridge},
  \doi{10.1017/9781139028868}.

\bibitemdeclare{inproceedings}{grohe}
\bibitem{grohe}
\bibinfo{author}{Martin \surnamestart Grohe\surnameend} (\bibinfo{year}{2021}):
  \emph{\bibinfo{title}{The logic of graph neural networks}}.
\newblock In: {\slshape \bibinfo{booktitle}{LICS '21: Proceedings of the 36th
  Annual ACM/IEEE Symposium on Logic in Computer Science}},
  \doi{10.1109/LICS52264.2021.9470677}.
\newblock \bibinfo{note}{Preprint arXiv:2104.14624 [cs.LG]}.

\bibitemdeclare{misc}{grohe2019linear}
\bibitem{grohe2019linear}
\bibinfo{author}{Martin \surnamestart Grohe\surnameend} \&
  \bibinfo{author}{Sandra \surnamestart Kiefer\surnameend}
  (\bibinfo{year}{2019}): \emph{\bibinfo{title}{A Linear Upper Bound on the
  {W}eisfeiler--{L}eman Dimension of Graphs of Bounded Genus}}.
\newblock \eprint{1904.07216}.

\bibitemdeclare{misc}{grohe2021logarithmic}
\bibitem{grohe2021logarithmic}
\bibinfo{author}{Martin \surnamestart Grohe\surnameend} \&
  \bibinfo{author}{Sandra \surnamestart Kiefer\surnameend}
  (\bibinfo{year}{2021}): \emph{\bibinfo{title}{Logarithmic
  {W}eisfeiler--{L}eman Identifies All Planar Graphs}}.
\newblock \eprint{2106.16218}.

\bibitemdeclare{misc}{grohe2019rankwidth}
\bibitem{grohe2019rankwidth}
\bibinfo{author}{Martin \surnamestart Grohe\surnameend} \&
  \bibinfo{author}{Daniel \surnamestart Neuen\surnameend}
  (\bibinfo{year}{2019}): \emph{\bibinfo{title}{Canonisation and Definability
  for Graphs of Bounded Rank Width}}.
\newblock \eprint{1901.10330}.

\bibitemdeclare{article}{GroheOttoLinearEquations}
\bibitem{GroheOttoLinearEquations}
\bibinfo{author}{Martin \surnamestart Grohe\surnameend} \&
  \bibinfo{author}{Martin \surnamestart Otto\surnameend}
  (\bibinfo{year}{2015}): \emph{\bibinfo{title}{Pebble Games and linear
  equations}}.
\newblock {\slshape \bibinfo{journal}{J. Symb. Log.}}
  \bibinfo{volume}{80}(\bibinfo{number}{3}), pp. \bibinfo{pages}{797--844},
  \doi{10.1017/jsl.2015.28}.

\bibitemdeclare{inproceedings}{GroheVerbitsky}
\bibitem{GroheVerbitsky}
\bibinfo{author}{Martin \surnamestart Grohe\surnameend} \&
  \bibinfo{author}{Oleg \surnamestart Verbitsky\surnameend}
  (\bibinfo{year}{2006}): \emph{\bibinfo{title}{Testing Graph Isomorphism in
  Parallel by Playing a Game}}.
\newblock In \bibinfo{editor}{Michele \surnamestart Bugliesi\surnameend},
  \bibinfo{editor}{Bart \surnamestart Preneel\surnameend},
  \bibinfo{editor}{Vladimiro \surnamestart Sassone\surnameend} \&
  \bibinfo{editor}{Ingo \surnamestart Wegener\surnameend}, editors: {\slshape
  \bibinfo{booktitle}{Automata, Languages and Programming, 33rd International
  Colloquium, {ICALP} 2006, Venice, Italy, July 10-14, 2006, Proceedings, Part
  {I}}}, {\slshape \bibinfo{series}{Lecture Notes in Computer Science}}
  \bibinfo{volume}{4051}, \bibinfo{publisher}{Springer}, pp.
  \bibinfo{pages}{3--14}, \doi{10.1007/11786986_2}.

\bibitemdeclare{article}{HeQiao}
\bibitem{HeQiao}
\bibinfo{author}{Xiaoyu \surnamestart He\surnameend} \&
  \bibinfo{author}{Youming \surnamestart Qiao\surnameend}
  (\bibinfo{year}{2021}): \emph{\bibinfo{title}{On the
  {Baer--Lov\'{a}sz--Tutte} construction of groups from graphs: {Isomorphism}
  types and homomorphism notions}}.
\newblock {\slshape \bibinfo{journal}{Eur. J. Combin.}} \bibinfo{volume}{98},
  p. \bibinfo{pages}{103404}, \doi{10.1016/j.ejc.2021.103404}.

\bibitemdeclare{article}{Heineken1974TheOO}
\bibitem{Heineken1974TheOO}
\bibinfo{author}{Hermann \surnamestart Heineken\surnameend} \&
  \bibinfo{author}{Hans \surnamestart Liebeck\surnameend}
  (\bibinfo{year}{1974}): \emph{\bibinfo{title}{The occurrence of finite groups
  in the automorphism group of nilpotent groups of class {$2$}}}.
\newblock {\slshape \bibinfo{journal}{Arch. Math. (Basel)}}
  \bibinfo{volume}{25}, pp. \bibinfo{pages}{8--16}, \doi{10.1007/BF01238631}.

\bibitemdeclare{article}{Hella1989}
\bibitem{Hella1989}
\bibinfo{author}{Lauri \surnamestart Hella\surnameend} (\bibinfo{year}{1989}):
  \emph{\bibinfo{title}{Definability hierarchies of generalized quantifiers}}.
\newblock {\slshape \bibinfo{journal}{Annals of Pure and Applied Logic}}
  \bibinfo{volume}{43}(\bibinfo{number}{3}), pp. \bibinfo{pages}{235 -- 271},
  \doi{10.1016/0168-0072(89)90070-5}.

\bibitemdeclare{article}{Hella1993}
\bibitem{Hella1993}
\bibinfo{author}{Lauri \surnamestart Hella\surnameend} (\bibinfo{year}{1996}):
  \emph{\bibinfo{title}{Logical Hierarchies in {PTIME}}}.
\newblock {\slshape \bibinfo{journal}{Information and Computation}}
  \bibinfo{volume}{129}(\bibinfo{number}{1}), pp. \bibinfo{pages}{1--19},
  \doi{10.1006/inco.1996.0070}.

\bibitemdeclare{article}{ImmermanPTime}
\bibitem{ImmermanPTime}
\bibinfo{author}{Neil \surnamestart Immerman\surnameend}
  (\bibinfo{year}{1986}): \emph{\bibinfo{title}{Relational Queries Computable
  in Polynomial Time}}.
\newblock {\slshape \bibinfo{journal}{Inf. Control.}}
  \bibinfo{volume}{68}(\bibinfo{number}{1-3}), pp. \bibinfo{pages}{86--104},
  \doi{10.1016/S0019-9958(86)80029-8}.

\bibitemdeclare{incollection}{ImmermanLander1990}
\bibitem{ImmermanLander1990}
\bibinfo{author}{Neil \surnamestart Immerman\surnameend} \&
  \bibinfo{author}{Eric \surnamestart Lander\surnameend}
  (\bibinfo{year}{1990}): \emph{\bibinfo{title}{Describing Graphs: A
  First-Order Approach to Graph Canonization}}.
\newblock In \bibinfo{editor}{Alan~L. \surnamestart Selman\surnameend}, editor:
  {\slshape \bibinfo{booktitle}{Complexity Theory Retrospective: In Honor of
  Juris Hartmanis on the Occasion of His Sixtieth Birthday, July 5, 1988}},
  \bibinfo{publisher}{Springer New York}, \bibinfo{address}{New York, NY}, pp.
  \bibinfo{pages}{59--81}, \doi{10.1007/978-1-4612-4478-3_5}.

\bibitemdeclare{article}{ETH}
\bibitem{ETH}
\bibinfo{author}{Russell \surnamestart Impagliazzo\surnameend},
  \bibinfo{author}{Ramamohan \surnamestart Paturi\surnameend} \&
  \bibinfo{author}{Francis \surnamestart Zane\surnameend}
  (\bibinfo{year}{2001}): \emph{\bibinfo{title}{Which Problems Have Strongly
  Exponential Complexity?}}
\newblock {\slshape \bibinfo{journal}{Journal of Computer and System Sciences}}
  \bibinfo{volume}{63}(\bibinfo{number}{4}), pp. \bibinfo{pages}{512--530},
  \doi{10.1006/jcss.2001.1774}.

\bibitemdeclare{article}{IvanyosQ19}
\bibitem{IvanyosQ19}
\bibinfo{author}{G{\'{a}}bor \surnamestart Ivanyos\surnameend} \&
  \bibinfo{author}{Youming \surnamestart Qiao\surnameend}
  (\bibinfo{year}{2019}): \emph{\bibinfo{title}{Algorithms Based on *-Algebras,
  and Their Applications to Isomorphism of Polynomials with One Secret, Group
  Isomorphism, and Polynomial Identity Testing}}.
\newblock {\slshape \bibinfo{journal}{{SIAM} J. Comput.}}
  \bibinfo{volume}{48}(\bibinfo{number}{3}), pp. \bibinfo{pages}{926--963},
  \doi{10.1137/18M1165682}.

\bibitemdeclare{article}{Kavitha}
\bibitem{Kavitha}
\bibinfo{author}{T.~\surnamestart Kavitha\surnameend} (\bibinfo{year}{2007}):
  \emph{\bibinfo{title}{Linear time algorithms for Abelian group isomorphism
  and related problems}}.
\newblock {\slshape \bibinfo{journal}{Journal of Computer and System Sciences}}
  \bibinfo{volume}{73}(\bibinfo{number}{6}), pp. \bibinfo{pages}{986 -- 996},
  \doi{10.1016/j.jcss.2007.03.013}.

\bibitemdeclare{inproceedings}{KayalNezhmetdinov}
\bibitem{KayalNezhmetdinov}
\bibinfo{author}{Neeraj \surnamestart Kayal\surnameend} \&
  \bibinfo{author}{Timur \surnamestart Nezhmetdinov\surnameend}
  (\bibinfo{year}{2009}): \emph{\bibinfo{title}{Factoring Groups Efficiently}}.
\newblock In \bibinfo{editor}{Susanne \surnamestart Albers\surnameend},
  \bibinfo{editor}{Alberto \surnamestart Marchetti-Spaccamela\surnameend},
  \bibinfo{editor}{Yossi \surnamestart Matias\surnameend},
  \bibinfo{editor}{Sotiris \surnamestart Nikoletseas\surnameend} \&
  \bibinfo{editor}{Wolfgang \surnamestart Thomas\surnameend}, editors:
  {\slshape \bibinfo{booktitle}{Automata, Languages and Programming}},
  \bibinfo{publisher}{Springer Berlin Heidelberg}, \bibinfo{address}{Berlin,
  Heidelberg}, pp. \bibinfo{pages}{585--596},
  \doi{10.1007/978-3-642-02927-1_49}.

\bibitemdeclare{inproceedings}{KieferMcKay}
\bibitem{KieferMcKay}
\bibinfo{author}{Sandra \surnamestart Kiefer\surnameend} \&
  \bibinfo{author}{Brendan~D. \surnamestart McKay\surnameend}
  (\bibinfo{year}{2020}): \emph{\bibinfo{title}{The Iteration Number of Colour
  Refinement}}.
\newblock In \bibinfo{editor}{Artur \surnamestart Czumaj\surnameend},
  \bibinfo{editor}{Anuj \surnamestart Dawar\surnameend} \&
  \bibinfo{editor}{Emanuela \surnamestart Merelli\surnameend}, editors:
  {\slshape \bibinfo{booktitle}{47th International Colloquium on Automata,
  Languages, and Programming, {ICALP} 2020, July 8-11, 2020, Saarbr{\"{u}}cken,
  Germany (Virtual Conference)}}, {\slshape \bibinfo{series}{LIPIcs}}
  \bibinfo{volume}{168}, \bibinfo{publisher}{Schloss Dagstuhl - Leibniz-Zentrum
  f{\"{u}}r Informatik}, pp. \bibinfo{pages}{73:1--73:19},
  \doi{10.4230/LIPIcs.ICALP.2020.73}.

\bibitemdeclare{article}{KieferPonomarenkoSchweitzer}
\bibitem{KieferPonomarenkoSchweitzer}
\bibinfo{author}{Sandra \surnamestart Kiefer\surnameend}, \bibinfo{author}{Ilia
  \surnamestart Ponomarenko\surnameend} \& \bibinfo{author}{Pascal
  \surnamestart Schweitzer\surnameend} (\bibinfo{year}{2019}):
  \emph{\bibinfo{title}{The {Weisfeiler--Leman} Dimension of Planar Graphs Is
  at Most 3}}.
\newblock {\slshape \bibinfo{journal}{J. ACM}}
  \bibinfo{volume}{66}(\bibinfo{number}{6}), \doi{10.1145/3333003}.

\bibitemdeclare{article}{KieferSchweitzerSelman}
\bibitem{KieferSchweitzerSelman}
\bibinfo{author}{Sandra \surnamestart Kiefer\surnameend},
  \bibinfo{author}{Pascal \surnamestart Schweitzer\surnameend} \&
  \bibinfo{author}{Erkal \surnamestart Selman\surnameend}
  (\bibinfo{year}{2022}): \emph{\bibinfo{title}{Graphs Identified by Logics
  with Counting}}.
\newblock {\slshape \bibinfo{journal}{{ACM} Trans. Comput. Log.}}
  \bibinfo{volume}{23}(\bibinfo{number}{1}), pp. \bibinfo{pages}{1:1--1:31},
  \doi{10.1145/3417515}.

\bibitemdeclare{article}{GILowPP}
\bibitem{GILowPP}
\bibinfo{author}{Johannes \surnamestart K{\"{o}}bler\surnameend},
  \bibinfo{author}{Uwe \surnamestart Sch{\"{o}}ning\surnameend} \&
  \bibinfo{author}{Jacobo \surnamestart Tor{\'{a}}n\surnameend}
  (\bibinfo{year}{1992}): \emph{\bibinfo{title}{Graph Isomorphism is Low for
  {PP}}}.
\newblock {\slshape \bibinfo{journal}{Comput. Complex.}} \bibinfo{volume}{2},
  pp. \bibinfo{pages}{301--330}, \doi{10.1007/BF01200427}.

\bibitemdeclare{article}{Ladner}
\bibitem{Ladner}
\bibinfo{author}{Richard~E. \surnamestart Ladner\surnameend}
  (\bibinfo{year}{1975}): \emph{\bibinfo{title}{On the Structure of Polynomial
  Time Reducibility}}.
\newblock {\slshape \bibinfo{journal}{J. ACM}}
  \bibinfo{volume}{22}(\bibinfo{number}{1}), p. \bibinfo{pages}{155–171},
  \doi{10.1145/321864.321877}.

\bibitemdeclare{inproceedings}{Gal09}
\bibitem{Gal09}
\bibinfo{author}{Fran\c{c}ois \surnamestart Le~Gall\surnameend}
  (\bibinfo{year}{2009}): \emph{\bibinfo{title}{Efficient Isomorphism Testing
  for a Class of Group Extensions}}.
\newblock In: {\slshape \bibinfo{booktitle}{Proc. 26th STACS}}, pp.
  \bibinfo{pages}{625--636}, \doi{10.4230/LIPIcs.STACS.2009.1830}.

\bibitemdeclare{misc}{GR16}
\bibitem{GR16}
\bibinfo{author}{Fran{\c{c}}ois \surnamestart Le~Gall\surnameend} \&
  \bibinfo{author}{David~J. \surnamestart Rosenbaum\surnameend}
  (\bibinfo{year}{2016}): \emph{\bibinfo{title}{On the Group and Color
  Isomorphism Problems}}.
\newblock \bibinfo{howpublished}{\arXiv{1609.08253}{[cs.CC]}}.

\bibitemdeclare{article}{LW12}
\bibitem{LW12}
\bibinfo{author}{Mark~L. \surnamestart Lewis\surnameend} \&
  \bibinfo{author}{James~B. \surnamestart Wilson\surnameend}
  (\bibinfo{year}{2012}): \emph{\bibinfo{title}{Isomorphism in expanding
  families of indistinguishable groups}}.
\newblock {\slshape \bibinfo{journal}{Groups - Complexity - Cryptology}}
  \bibinfo{volume}{4}(\bibinfo{number}{1}), pp. \bibinfo{pages}{73--110},
  \doi{10.1515/gcc-2012-0008}.

\bibitemdeclare{article}{Lindstrom}
\bibitem{Lindstrom}
\bibinfo{author}{P.~\surnamestart Lindstrom\surnameend} (\bibinfo{year}{1966}):
  \emph{\bibinfo{title}{First Order Predicate Logic with Generalized
  Quantifiers}}.
\newblock {\slshape \bibinfo{journal}{Theoria}}
  \bibinfo{volume}{32}(\bibinfo{number}{3}), pp. \bibinfo{pages}{186--195},
  \doi{10.1111/j.1755-2567.1966.tb00600.x}.

\bibitemdeclare{article}{Mekler}
\bibitem{Mekler}
\bibinfo{author}{Alan~H. \surnamestart Mekler\surnameend}
  (\bibinfo{year}{1981}): \emph{\bibinfo{title}{Stability of Nilpotent Groups
  of Class 2 and Prime Exponent}}.
\newblock {\slshape \bibinfo{journal}{The Journal of Symbolic Logic}}
  \bibinfo{volume}{46}(\bibinfo{number}{4}), pp. \bibinfo{pages}{781--788},
  \doi{10.2307/2273227}.
\newblock \urlprefix\url{http://www.jstor.org/stable/2273227}.

\bibitemdeclare{inproceedings}{MillerTarjan}
\bibitem{MillerTarjan}
\bibinfo{author}{Gary~L. \surnamestart Miller\surnameend}
  (\bibinfo{year}{1978}): \emph{\bibinfo{title}{On the {$n^{\log n}$}
  Isomorphism Technique (A Preliminary Report)}}.
\newblock In: {\slshape \bibinfo{booktitle}{Proceedings of the Tenth Annual ACM
  Symposium on Theory of Computing}}, \bibinfo{series}{STOC '78},
  \bibinfo{publisher}{Association for Computing Machinery},
  \bibinfo{address}{New York, NY, USA}, pp. \bibinfo{pages}{51--58},
  \doi{10.1145/800133.804331}.

\bibitemdeclare{article}{Mostowski1957}
\bibitem{Mostowski1957}
\bibinfo{author}{Andrzej \surnamestart Mostowski\surnameend}
  (\bibinfo{year}{1957}): \emph{\bibinfo{title}{On a generalization of
  quantifiers}}.
\newblock {\slshape \bibinfo{journal}{Fundamenta Mathematicae}}
  \bibinfo{volume}{44}(\bibinfo{number}{1}), pp. \bibinfo{pages}{12--36},
  \doi{10.4064/fm-44-1-12-36}.
\newblock \urlprefix\url{http://eudml.org/doc/213418}.

\bibitemdeclare{article}{FiniteGroupsFOL}
\bibitem{FiniteGroupsFOL}
\bibinfo{author}{Andr\'{e} \surnamestart Nies\surnameend} \&
  \bibinfo{author}{Katrin \surnamestart Tent\surnameend}
  (\bibinfo{year}{2017}): \emph{\bibinfo{title}{Describing finite groups by
  short first-order sentences}}.
\newblock {\slshape \bibinfo{journal}{Israel J. Math.}}
  \bibinfo{volume}{221}(\bibinfo{number}{1}), pp. \bibinfo{pages}{85--115},
  \doi{10.1007/s11856-017-1563-2}.

\bibitemdeclare{inproceedings}{QST11}
\bibitem{QST11}
\bibinfo{author}{Youming \surnamestart Qiao\surnameend},
  \bibinfo{author}{Jayalal M.~N. \surnamestart Sarma\surnameend} \&
  \bibinfo{author}{Bangsheng \surnamestart Tang\surnameend}
  (\bibinfo{year}{2011}): \emph{\bibinfo{title}{On Isomorphism Testing of
  Groups with Normal {Hall} Subgroups}}.
\newblock In: {\slshape \bibinfo{booktitle}{Proc. 28th STACS}}, pp.
  \bibinfo{pages}{567--578}, \doi{10.4230/LIPIcs.STACS.2011.567}.

\bibitemdeclare{misc}{Rosenbaum2013BidirectionalCD}
\bibitem{Rosenbaum2013BidirectionalCD}
\bibinfo{author}{David~J. \surnamestart Rosenbaum\surnameend}
  (\bibinfo{year}{2013}): \emph{\bibinfo{title}{Bidirectional Collision
  Detection and Faster Deterministic Isomorphism Testing}}.
\newblock \bibinfo{howpublished}{\arXiv{1304.3935}{[cs.DS]}}.

\bibitemdeclare{inproceedings}{Rossman2009EhrenfeuchtFrassGO}
\bibitem{Rossman2009EhrenfeuchtFrassGO}
\bibinfo{author}{Benjamin \surnamestart Rossman\surnameend}
  (\bibinfo{year}{2009}):
  \emph{\bibinfo{title}{Ehrenfeucht--{F}ra{\"{\i}}ss{\'{e}} Games on Random
  Structures}}.
\newblock In \bibinfo{editor}{Hiroakira \surnamestart Ono\surnameend},
  \bibinfo{editor}{Makoto \surnamestart Kanazawa\surnameend} \&
  \bibinfo{editor}{Ruy J. G.~B. \surnamestart de~Queiroz\surnameend}, editors:
  {\slshape \bibinfo{booktitle}{Logic, Language, Information and Computation,
  16th International Workshop, WoLLIC 2009, Tokyo, Japan, June 21-24, 2009.
  Proceedings}}, {\slshape \bibinfo{series}{Lecture Notes in Computer Science}}
  \bibinfo{volume}{5514}, \bibinfo{publisher}{Springer}, pp.
  \bibinfo{pages}{350--364}, \doi{10.1007/978-3-642-02261-6_28}.

\bibitemdeclare{techreport}{Savage}
\bibitem{Savage}
\bibinfo{author}{C.~\surnamestart Savage\surnameend} (\bibinfo{year}{1980}):
  \emph{\bibinfo{title}{An {$O(n^2)$} Algorithm for Abelian Group
  Isomorphism}}.
\newblock \bibinfo{type}{Technical Report}, \bibinfo{institution}{North
  Carolina State University}.

\bibitemdeclare{article}{Schoning}
\bibitem{Schoning}
\bibinfo{author}{Uwe \surnamestart Sch{\"{o}}ning\surnameend}
  (\bibinfo{year}{1988}): \emph{\bibinfo{title}{Graph isomorphism is in the low
  hierarchy}}.
\newblock {\slshape \bibinfo{journal}{Journal of Computer and System Sciences}}
  \bibinfo{volume}{37}(\bibinfo{number}{3}), pp. \bibinfo{pages}{312 -- 323},
  \doi{10.1016/0022-0000(88)90010-4}.

\bibitemdeclare{inproceedings}{VardiPTime}
\bibitem{VardiPTime}
\bibinfo{author}{Moshe~Y. \surnamestart Vardi\surnameend}
  (\bibinfo{year}{1982}): \emph{\bibinfo{title}{The Complexity of Relational
  Query Languages (Extended Abstract)}}.
\newblock In \bibinfo{editor}{Harry~R. \surnamestart Lewis\surnameend},
  \bibinfo{editor}{Barbara~B. \surnamestart Simons\surnameend},
  \bibinfo{editor}{Walter~A. \surnamestart Burkhard\surnameend} \&
  \bibinfo{editor}{Lawrence~H. \surnamestart Landweber\surnameend}, editors:
  {\slshape \bibinfo{booktitle}{Proceedings of the 14th Annual {ACM} Symposium
  on Theory of Computing, May 5-7, 1982, San Francisco, California, {USA}}},
  \bibinfo{publisher}{{ACM}}, pp. \bibinfo{pages}{137--146},
  \doi{10.1145/800070.802186}.

\bibitemdeclare{article}{Vikas}
\bibitem{Vikas}
\bibinfo{author}{Narayan \surnamestart Vikas\surnameend}
  (\bibinfo{year}{1996}): \emph{\bibinfo{title}{An {$O(n)$} Algorithm for
  Abelian {$p$}-Group Isomorphism and an {$O(n \log n)$} Algorithm for Abelian
  Group Isomorphism}}.
\newblock {\slshape \bibinfo{journal}{Journal of Computer and System Sciences}}
  \bibinfo{volume}{53}(\bibinfo{number}{1}), pp. \bibinfo{pages}{1--9},
  \doi{10.1006/jcss.1996.0045}.

\bibitemdeclare{article}{WilsonDirectProductsArxiv}
\bibitem{WilsonDirectProductsArxiv}
\bibinfo{author}{James~B. \surnamestart Wilson\surnameend}
  (\bibinfo{year}{2012}): \emph{\bibinfo{title}{Existence, algorithms, and
  asymptotics of direct product decompositions, {I}}}.
\newblock {\slshape \bibinfo{journal}{Groups - Complexity - Cryptology}}
  \bibinfo{volume}{4}(\bibinfo{number}{1}), \doi{10.1515/gcc-2012-0007}.

\bibitemdeclare{article}{WilsonSubgroupProfiles}
\bibitem{WilsonSubgroupProfiles}
\bibinfo{author}{James~B. \surnamestart Wilson\surnameend}
  (\bibinfo{year}{2019}): \emph{\bibinfo{title}{The Threshold for Subgroup
  Profiles to Agree is Logarithmic}}.
\newblock {\slshape \bibinfo{journal}{Theory of Computing}}
  \bibinfo{volume}{15}(\bibinfo{number}{19}), pp. \bibinfo{pages}{1--25},
  \doi{10.4086/toc.2019.v015a019}.

\bibitemdeclare{article}{ZKT}
\bibitem{ZKT}
\bibinfo{author}{V.~N. \surnamestart Zemlyachenko\surnameend},
  \bibinfo{author}{N.~M. \surnamestart Korneenko\surnameend} \&
  \bibinfo{author}{R.~I. \surnamestart Tyshkevich\surnameend}
  (\bibinfo{year}{1985}): \emph{\bibinfo{title}{Graph isomorphism problem}}.
\newblock {\slshape \bibinfo{journal}{J. Soviet Math.}}
  \bibinfo{volume}{29}(\bibinfo{number}{4}), pp. \bibinfo{pages}{1426--1481},
  \doi{10.1007/BF02104746}.

\end{thebibliography}

\end{document}